\documentclass{llncs}

\usepackage{multicol}
\usepackage{amsmath}
\usepackage{amsfonts}
\usepackage{amssymb}
\usepackage{graphicx}
\usepackage{epic}
\usepackage{eepic}
\usepackage{epsfig,float}
\usepackage{pdfsync}

\renewcommand{\le}{\leqslant}
\renewcommand{\ge}{\geqslant}

\newcommand{\eps}{\varepsilon}
\newcommand{\emp}{\emptyset}

\newcommand{\Sig}{\Sigma}
\newcommand{\sig}{\sigma}
\newcommand{\noin}{\noindent}

\newcommand{\bi}{\begin{itemize}}
\newcommand{\ei}{\end{itemize}}
\newcommand{\be}{\begin{enumerate}}
\newcommand{\ee}{\end{enumerate}}
\newcommand{\bd}{\begin{description}}
\newcommand{\ed}{\end{description}}
\newcommand{\bq}{\begin{quote}}
\newcommand{\eq}{\end{quote}}
\newcommand{\txt}[1]{\mbox{ #1 }}
\newcommand{\defeq}{\stackrel{\rm def}{=}}

\newcommand{\timg}{\mathop{\mbox{rng}}}
\newcommand{\tdom}{\mathop{\mbox{dom}}}

\newcommand{\cA}{{\mathcal A}}
\newcommand{\cB}{{\mathcal B}}
\newcommand{\cC}{{\mathcal C}}
\newcommand{\cD}{{\mathcal D}}

\newcommand{\cM}{{\mathcal M}}
\newcommand{\cN}{{\mathcal N}}

\newcommand{\cT}{{\mathcal T}}

\newcommand{\raL}{{\hspace{.1cm}{\sim_L} \hspace{.1cm}}}
\newcommand{\lraL}{{\hspace{.1cm}{\approx_L} \hspace{.1cm}}}
\newcommand{\seq}{{\mathfrak{s}}}
\newcommand{\tpath}{P}

\newcommand{\Bsf}{\mathbf{B}_{\mathrm{sf}}} 
\newcommand{\Bbf}{\mathbf{B}_{\mathrm{bf}}} 
\newcommand{\Bff}{\mathbf{B}_{\mathrm{ff}}}
\newcommand{\Vsf}{\mathbf{W}^{\le 5}_{\mathrm{sf}}} 
\newcommand{\Wsf}{\mathbf{W}^{\ge 6}_{\mathrm{sf}}} 
\newcommand{\Vbf}{\mathbf{W}^{\le 5}_{\mathrm{bf}}} 
\newcommand{\Wbf}{\mathbf{W}^{\ge 6}_{\mathrm{bf}}} 
\newcommand{\Wff}{\mathbf{W}_{\mathrm{ff}}} 
\newcommand{\Gsf}{\mathbf{G}^{\le 5}_{\mathrm{sf}}} 
\newcommand{\Hsf}{\mathbf{G}^{\ge 6}_{\mathrm{sf}}}
\newcommand{\Gbf}{\mathbf{G}^{\le 5}_{\mathrm{bf}}} 
\newcommand{\Hbf}{\mathbf{G}^{\ge 6}_{\mathrm{bf}}}
\newcommand{\Hff}{\mathbf{G}_{\mathrm{ff}}} 
\newcommand{\Uf}{\mathbf{U}}

\newcommand{\bsf}{{\mathrm{b}_{\mathrm{sf}}}} 
\newcommand{\bbf}{{\mathrm{b}_{\mathrm{bf}}}} 
\newcommand{\bff}{{\mathrm{b}_{\mathrm{ff}}}} 
\newcommand{\vsf}{{\mathrm{w}^{\le 5}_{\mathrm{sf}}}} 
\newcommand{\wsf}{{\mathrm{w}^{\ge 6}_{\mathrm{sf}}}}
\newcommand{\vbf}{{\mathrm{w}^{\le 5}_{\mathrm{bf}}}} 
\newcommand{\wbf}{{\mathrm{w}^{\ge 6}_{\mathrm{bf}}}} 
\newcommand{\wff}{{\mathrm{w}_{\mathrm{ff}}}}

\titlerunning{Syntactic Complexity of Prefix-, Suffix-, Bifix-, and Factor-Free Regular Languages}
\authorrunning{J.~Brzozowski, B.~Li, and Y.~Ye}

\title{Syntactic Complexity of Prefix-, Suffix-, Bifix-, and Factor-Free Regular Languages
\thanks{This work was supported by the Natural Sciences and Engineering Research Council of Canada under grant No.~OGP0000871 and  a Postgraduate Scholarship, and by a Graduate Award from the Department of Computer Science, University of Toronto.
}}

\author{Janusz~Brzozowski\inst{1}, Baiyu Li\inst{1}, \and Yuli Ye\inst{2}}

\institute{David R. Cheriton School of Computer Science, University of Waterloo \\
Waterloo, ON, Canada N2L 3G1\\
email: \email{\{brzozo, b5li\}@uwaterloo.ca}
\and
Department of Computer Science, University of Toronto\\
 Toronto, ON,  Canada M5S 3G4\\
email: \email{y3ye@cs.toronto.edu}
}

\begin{document}

\maketitle

\begin{abstract}
The syntactic complexity of a regular language is the cardinality of its syntactic semigroup.
The syntactic complexity of a subclass of the class of regular languages is the maximal syntactic complexity of languages in that class, taken as a function of the state complexity $n$ of these languages.
We study the syntactic complexity of  prefix-, suffix-, bifix-, and factor-free regular languages.
We prove that $n^{n-2}$ is a tight upper bound for prefix-free regular languages. We present properties of the syntactic semigroups of suffix-, bifix-, and factor-free regular languages, conjecture tight upper bounds on their size to be $(n-1)^{n-2}+(n-2)$, $(n-1)^{n-3} + (n-2)^{n-3} + (n-3)2^{n-3}$, and $(n-1)^{n-3} + (n-3)2^{n-3} + 1$, respectively, and exhibit languages with these syntactic complexities.
\end{abstract}

\noin{\bf keyword}
bifix-free, factor-free, finite automaton, monoid, prefix-free, regular language, reversal, semigroup,  suffix-free, syntactic complexity 

\section{Introduction}

A language is \emph{prefix-free} (respectively, \emph{suffix-free}, \emph{factor-free}) if it does not contain any pair of words such that one is a proper prefix (respectively, suffix, factor) of the other. It is \emph{bifix-free} if it is both prefix- and suffix-free.
We refer to  prefix-, suffix-, bifix-, and factor-free languages as \emph{free} languages.
Nontrivial prefix-, suffix-, bifix-, and factor-free languages are also known as prefix, suffix, bifix, and infix codes~\cite{BPR09,Shy01}, respectively and, have many applications in areas such as cryptography, data compression, and information processing.

The \emph{state complexity} of a regular language is the  number of states in the minimal deterministic finite automaton (DFA) recognizing that language. 
An equivalent notion is that of \emph{quotient complexity,} which is the number of left quotients of the language.
State complexity of regular operations has been studied quite extensively: for surveys of this topic and lists of references we refer the reader to~\cite{Brz09,Yu01}. 
With regard to free regular languages,  Han, Salomaa and Wood~\cite{HSW09} examined  prefix-free regular languages, and  Han and Salomaa~\cite{HS09} studied suffix-free regular languages. 
Bifix- and factor-free regular languages were studied by Brzozowski, Jir\'askov\'a, Li, and Smith~\cite{BJLS11}. 

The notion of quotient complexity can be derived from the Nerode right congruence~\cite{Ner58}, 
while the Myhill congruence~\cite{Myh57} leads to the syntactic semigroup of a language and to its \emph{syntactic complexity}, which is the cardinality of the syntactic semigroup.
It was pointed out in~\cite{BrYe11} that syntactic complexity can be very different for regular languages with the same quotient complexity.
Thus, for a fixed $n$, languages with quotient complexity $n$ may possibly be distinguished by their syntactic complexities.

In contrast to state complexity, syntactic complexity has not received much attention. In 1970 Maslov~\cite{Mas70} dealt with the problem of generators of the semigroup of all transformations in the setting of finite automata. In 2003--2004, Holzer and K\"onig~\cite{HoKo04}, and independently, Krawetz, Lawrence and Shallit~\cite{KLS03} studied the syntactic complexity of languages with unary and binary alphabets. 
In 2010 Brzozowski and Ye~\cite{BrYe11} examined the syntactic complexity of ideal and closed regular languages, and in 2011 Brzozowski and Li~\cite{BL11} studied the syntactic complexity of star-free languages. 
Here, we deal with the syntactic complexity of  prefix-, suffix-, bifix-, and factor-free regular languages, and their complements.

Basic definitions and facts are stated in Sections~\ref{sec:trans} and~\ref{sec:complexity}. In Section~\ref{sec:pf} we obtain a tight upper bound on the syntactic complexity of prefix-free regular languages. In Sections~\ref{sec:sf}--\ref{sec:ff} we study the syntactic complexity of suffix-, bifix-, and factor-free regular languages, respectively.  We state conjectures about tight upper bounds for these classes, and exhibit languages in these classes that have large syntactic complexities.   In Section~\ref{sec:rev} we show that the upper bounds on the quotient complexity of reversal of prefix-, suffix-, bifix-, and factor-free regular languages can be met by our languages with largest syntactic complexities. Section~\ref{sec:cl} concludes the paper.

\section{Transformations}\label{sec:trans}

A {\em transformation} of a set $Q$ is a mapping of $Q$ into itself. In this paper we consider only transformations of finite sets, and we assume without loss of generality  that $Q = \{1,2,\ldots, n\}$. Let $t$ be a transformation of $Q$. If  $i \in Q$, then $it$ is the {\it image} of $i$ under $t$.  If $X$ is a subset of $Q$, then $Xt = \{it \mid i \in X\}$, and the {\em restriction} of $t$ to $X$, denoted by $t|_X$, is a mapping from $X$ to $Xt$ such that $it|_X = it$ for all $i \in X$. The {\em composition} of two transformations $t_1$ and $t_2$ of $Q$ is a transformation $t_1 \circ t_2$ such that $i (t_1 \circ t_2) = (i t_1) t_2$ for all $i \in Q$. We usually drop the composition operator ``$\circ$'' and write $t_1t_2$ for short. 
An arbitrary transformation can be written in the form
\begin{equation*}\label{eq:transmatrix}
t=\left( \begin{array}{ccccc}
1 & 2 &   \cdots &  n-1 & n \\
i_1 & i_2 &   \cdots &  i_{n-1} & i_n
\end{array} \right ),
\end{equation*}
where $i_k = kt$,  $1\le k\le n$, and $i_k\in Q$. The {\em domain} $\tdom(t)$ of $t$ is $Q.$
The {\em range} $\timg(t)$ of $Q$ under $t$ is the set
$\timg(t) = Q t.$ We also use the notation $t = [i_1,i_2,\ldots,i_n]$ for the transformation $t$ above. 

A \emph{permutation} of $Q$ is a mapping of $Q$ \emph{onto} itself. In other words,  a permutation $\pi$ of $Q$ is a transformation where $\timg(\pi) = Q$. 
The \emph{identity} transformation maps each element to itself, that is, $it=i$ for $i=1,\ldots,n$.
A transformation $t$ contains a \emph{cycle} of length $k$ if there exist pairwise different elements $i_1,\ldots,i_k$ such that
$i_1t=i_2, i_2t=i_3,\ldots, i_{k-1}t=i_k$, and $i_kt=i_1$.
A cycle is denoted by $(i_1,i_2,\ldots,i_k)$.
For $i<j$, a \emph{transposition} is the cycle $(i,j)$, and $(i,i)$ is the identity.
A \emph{singular} transformation, denoted by $i\choose j$, has $it=j$ and $ht=h$ for all $h\neq i$, and $i \choose i$ is the identity.
A~\emph{constant} transformation,  denoted by $Q \choose j$, has $it=j$ for all $i$.

The set of all transformations of a set $Q$, denoted by $\cT_Q$, is a finite monoid. The set of all permutations of $Q$ is a group, denoted by $\mathfrak{S}_Q$ and called the \emph{symmetric group} of degree $n$. 
It was shown in~\cite{Hoy1895,Pic38} that two generators are sufficient to generate the symmetric group of degree $n$. 
In 1935 Piccard~\cite{Pic35} proved that three transformations of $Q$ are sufficient to generate the monoid $\cT_Q$. In the same year, Eilenberg showed that fewer than three generators are not possible, as reported by Sierpi\'nski~\cite{Sie35}. We refer the reader to the book of Ganyushkin and Mazorchuk~\cite{GaMa09} for a detailed discussion of finite transformation semigroups. The following are well-known facts about generators of $\mathfrak{S}_Q$ and $\cT_Q$: 

\begin{theorem}[Permutations,~\cite{Hoy1895,Pic38}]
\label{thm:piccard}
The symmetric group $\mathfrak{S}_Q$ of size $n!$ can be generated by any cyclic
permutation of $n$ elements together with any transposition. In particular, $\mathfrak{S}_Q$ can be generated by
$c=(1,2,\ldots, n)$~and~$t~=~(1,2)$.
\end{theorem}

\begin{theorem}[Transformations,~\cite{Pic35}]
\label{thm:salomaa}
The complete transformation monoid $\cT_Q$ of size $n^n$ can be generated by any cyclic
permutation of $n$ elements together with a transposition and a ``returning'' transformation $r={n \choose 1}$. In particular, $\cT_Q$ can be generated by $c=(1,2,\ldots, n)$,  $t=(1,2)$ and $r={n \choose 1}$.
\end{theorem}

\section{Quotient Complexity and Syntactic Complexity}\label{sec:complexity} 

If $\Sig$ is a non-empty finite alphabet, then $\Sig^*$ is the free monoid generated by $\Sig$, and $\Sig^+$ is the free semigroup generated by $\Sig$.  A \emph{word} is any element of $\Sig^*$, and the empty word is $\eps$. The length of a word $w\in \Sig^*$ is $|w|$. A \emph{language} over $\Sig$ is any subset of $\Sig^*$. 
If $w=uxv$ for some $u,x,v\in\Sigma^*$, then $u$ is a {\em prefix\/} of $w$, $v$ is a {\em suffix\/} of $w$, and $x$ is a {\em factor\/} of $w$. Both $u$ and $v$ are also factors of $w$. 
A~{\em proper} prefix (suffix, factor) of $w$ is a prefix (suffix, factor) of $w$ other than~$w$.

The \emph{left quotient}, or simply \emph{quotient,} of a language $L$ by a word $w$ is  the language $L_w=\{x\in \Sig^*\mid wx\in L \}$. 
For any  $L\subseteq \Sig^*$, the \emph{Nerode right congruence}~\cite{Ner58} $\raL$ of $L$ is defined as follows: 
\begin{equation*}
x \raL y \mbox{ if and only if } xv\in L  \Leftrightarrow yv\in L, \mbox { for all } v\in\Sig^*.
\end{equation*}
Clearly, $L_x=L_y$ if and only if $x\raL y$.
Thus each equivalence class of this right congruence corresponds to a distinct quotient of $L$.

The \emph{Myhill congruence}~\cite{Myh57} $\lraL$ of $L$ is defined as follows:
\begin{equation*}
x \lraL y \mbox{ if and only if } uxv\in L  \Leftrightarrow uyv\in L\mbox { for all } u,v\in\Sig^*.
\end{equation*}
This congruence is also known as the \emph{syntactic congruence} of $L$.
The quotient set $\Sig^+/ \lraL$ of equivalence classes of the relation $\lraL$ is a semigroup called the \emph{syntactic semigroup} of $L$, and 
$\Sig^*/ \lraL$ is the \emph{syntactic monoid} of~$L$. 
The \emph{syntactic complexity} $\sig(L)$ of $L$ is the cardinality of its syntactic semigroup.
The \emph{monoid complexity} $\mu(L)$ of $L$ is the cardinality of its syntactic monoid.
If the  equivalence class containing $\eps$ is a singleton in the syntactic monoid, then $\sig(L)=\mu(L)-1$; otherwise, $\sig(L)=\mu(L)$.

A~\emph{deterministic finite automaton} (DFA) is a quintuple $\cA=(Q, \Sig, \delta, q_1,F)$, where 
$Q$ is a finite, non-empty set of \emph{states}, $\Sig$ is a finite non-empty \emph{alphabet}, $\delta:Q\times \Sig\to Q$ is the \emph{transition function}, $q_1\in Q$ is the \emph{initial state}, and $F\subseteq Q$ is the set of \emph{accepting states}. We extend $\delta$ to $Q \times \Sig^*$ in the usual way.
The DFA $\cA$ accepts a word $w \in \Sigma^*$ if ${\delta}(q_1,w)\in F$. 
The set of all words {\it accepted} by $\cA$ is $L(\cA)$. 
By the \emph{language of a state} $q$ of $\cA$ 
we mean the language accepted 
by the DFA $(Q,\Sigma,\delta,q,F)$. 
A state is \emph{empty} if its language is empty.

Let $L$ be a regular language. 
The \emph{quotient DFA} of $L$ is 
$\cA=(Q, \Sig, \delta, q_1,F)$, where $Q=\{L_w\mid w\in\Sig^*\}$, $\delta(L_w,a)=L_{wa}$, 
$q_1=L_\eps=L$,  $F=\{L_w \mid \eps \in L_w\}$.
The number  $\kappa(L)$ of distinct quotients of $L$ is the \emph{quotient complexity} of $L$. 
The quotient DFA of $L$ is the minimal DFA accepting $L$, and so quotient complexity is the same as  state complexity, but there are advantages to using quotients~\cite{Brz09}.

In terms of automata, each equivalence class $[w]_{\raL}$ of $\raL$ is the set of all words $w$ that take the automaton to the same state from the initial state, and each equivalence class $[w]_\lraL$ of $\lraL$ is the set of all words that perform the same transformation on the set of states~\cite{McNP71}.
In terms of quotients, $[w]_{\raL}$ is the set of words $w$ that can be followed by the same quotient $L_w$. 

Let $\cA = (Q, \Sig, \delta, q_1, F)$ be a DFA. For each word $w \in \Sig^*$, the transition function for $w$ defines a transformation $t_w$ of $Q$ by the word $w$: for all $i \in Q$, 
$it_w \defeq \delta(i, w).$ 
The set $T_{\cA}$ of all such transformations by non-empty words forms a subsemigroup of $\cT_Q$, called the \emph{transition semigroup} of $\cA$~\cite{Pin97}. 
Conversely, we can use a set  $\{t_a \mid a \in \Sig\}$ of transformations to define $\delta$, and so the DFA $\cA$. When the context is clear we simply write $a = t$, where $t$ is a transformation of $Q$, to mean that the transformation performed by $a \in \Sig$ is~$t$.

If  $\cA$ is the quotient DFA of $L$, then $T_{\cA}$ is isomorphic to the syntactic semigroup $T_L$ of $L$~\cite{McNP71}, and we represent elements of $T_L$ by transformations in~$T_{\cA}$. 

We attempt to obtain tight upper bounds on the syntactic complexity $\sigma(L) = |T_L|$ of $L$ as a function of the quotient complexity $\kappa(L)$ of $L$.
First we consider the syntactic complexity of regular languages over a unary alphabet, where the  concepts prefix-, suffix-, bifix-, and factor-free, coincide. So we may consider only unary prefix-free regular languages $L$ with quotient complexity $\kappa(L) = n$. When $n = 1$, the only prefix-free language is $L = \emptyset$ with $\sigma(L) = 1$. For $n \ge 2$, a prefix-free language $L$ must be a singleton, $L = \{a^{n-2}\}$. The syntactic semigroup $T_L$ of $L$ consists of $n-1$ transformations $t_w$ by words $w = a^i$, where $1 \le i \le n-1$. Thus we have 

\begin{proposition}[Unary Free Regular Languages]
If $L$ is a unary free regular language with $\kappa(L) = n \ge 2$, then $\sigma(L) = n-1$.
\end{proposition}

The tight upper bound for regular unary languages~\cite{HoKo04} is $n$. 

We assume that $|\Sig| \ge 2$ in the following sections. 
Since the syntactic semigroup of a language is the same as that of its complement, we deal only with prefix-, suffix-, bifix-, and factor-free languages. All the syntactic complexity results, however, apply also to the complements of these languages.
\goodbreak

\section{Prefix-Free Regular Languages}\label{sec:pf}
To simplify notation we write $\eps$ for the language $\{\eps\}$. Recall that a regular language $L$ is prefix-free if and only it has exactly one accepting quotient, and that quotient is $\eps$~\cite{HSW09}.

\begin{theorem}[Prefix-Free Regular Languages]
\label{thm:prefix-free}
If $L$ is regular and prefix-free with $\kappa(L)=n\ge 2$, then $\sig(L)\le n^{n-2}$. Moreover, this bound is tight 
for $n=2$  if $|\Sig|\ge 1$, for $n=3$  if $|\Sig|\ge 2$, for $n=4$ if $|\Sig|\ge 4$, and for $n\ge 5$
if $|\Sig|\ge n+1$.
\end{theorem}
\begin{proof}
If $L$ is prefix-free, the only accepting quotient of $L$ is $\eps$. Thus $L$ also has the empty quotient, since $\eps_a = \emptyset$ for  $a \in \Sig$. 
Let $\cA = (Q, \Sig, \delta, 1, \{n-1\})$ be the quotient DFA of $L$, where, without loss of generality,  $n-1 \in Q$ is the only accepting state, and  $n \in Q$ is the empty state. For any transformation $t \in T_L$, $(n-1) t =n t = n$. Thus we have $\sig(L)\le n^{n-2}$.

The only prefix-free regular language for $n=1$ is $L=\emp$ with $\sig(L)=1$; here the bound $n^{n-2}$ does not apply. 
For $n=2$ and $\Sig=\{a\}$, the language $L=\eps$ meets the bound.
For $n=3$ and $\Sig=\{a,b\}$,  $L=b^*a$ meets the bound.
For $n\ge 4$, let $\cA_n=(\{1, 2,\ldots, n\}, \{a,b,c,d_1,d_2,\ldots,d_{n-2}\},\delta,1,\{n-1\})$, where $a={{n-1} \choose n} (1,2,\ldots,n-2)$, $b={{n-1} \choose n} (1,2)$, $c={{n-1} \choose n} {n-2\choose 1}$, and $d_i={{n-1} \choose n} {i\choose n-1}$ for $i=1,2,\ldots,n-2$.
DFA $\cA_6$ is shown in Fig.~\ref{fig:PrFree}, 
where $\Gamma = \{d_1,d_2,\ldots,d_{n-2}\}$.
For $n=4$,  input $a$ coincides with $b$; hence only $4$ inputs are needed. 

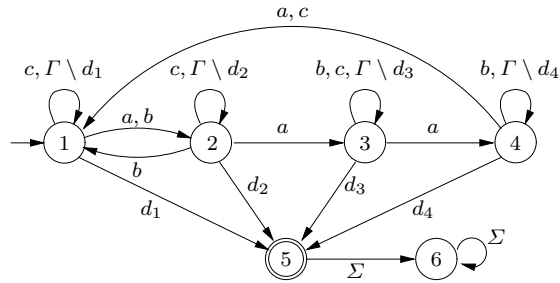
\begin{figure}[hbt]
\begin{center}
\setlength{\unitlength}{0.00056868in}
\begingroup\makeatletter\ifx\SetFigFont\undefined%
\gdef\SetFigFont#1#2#3#4#5{%
  \reset@font\fontsize{#1}{#2pt}%
  \fontfamily{#3}\fontseries{#4}\fontshape{#5}%
  \selectfont}%
\fi\endgroup%
{\renewcommand{\dashlinestretch}{30}
\begin{picture}(4890,2694)(0,-10)
\put(3207,78){\makebox(0,0)[b]{\smash{{\SetFigFont{8}{9.6}{\familydefault}{\mddefault}{\updefault}$\Sig$}}}}
\put(507,1270){\makebox(0,0)[b]{\smash{{\SetFigFont{8}{9.6}{\familydefault}{\mddefault}{\updefault}1}}}}
\put(1866,1338){\ellipse{382}{382}}
\put(1880,1270){\makebox(0,0)[b]{\smash{{\SetFigFont{8}{9.6}{\familydefault}{\mddefault}{\updefault}2}}}}
\put(3296,1343){\ellipse{382}{382}}
\put(3297,1270){\makebox(0,0)[b]{\smash{{\SetFigFont{8}{9.6}{\familydefault}{\mddefault}{\updefault}3}}}}
\put(4691,1343){\ellipse{382}{382}}
\put(4692,1270){\makebox(0,0)[b]{\smash{{\SetFigFont{8}{9.6}{\familydefault}{\mddefault}{\updefault}4}}}}
\put(2564,258){\ellipse{382}{382}}
\put(2564,258){\ellipse{324}{324}}
\put(2564,190){\makebox(0,0)[b]{\smash{{\SetFigFont{8}{9.6}{\familydefault}{\mddefault}{\updefault}5}}}}
\put(3957,259){\ellipse{382}{382}}
\put(3959,190){\makebox(0,0)[b]{\smash{{\SetFigFont{8}{9.6}{\familydefault}{\mddefault}{\updefault}6}}}}
\put(4285.677,317.559){\arc{274.198}{3.3655}{8.5574}}
\blacken\path(4320.693,213.226)(4197.000,213.000)(4306.247,154.991)(4320.693,213.226)
\put(4512,393){\makebox(0,0)[b]{\smash{{\SetFigFont{8}{9.6}{\familydefault}{\mddefault}{\updefault}$\Sig$}}}}
\put(1182.000,66.750){\arc{2812.500}{4.3527}{5.0721}}
\blacken\path(1553.620,1391.800)(1677.000,1383.000)(1572.277,1448.826)(1553.620,1391.800)
\put(1182.000,2609.250){\arc{2812.500}{1.2111}{1.9305}}
\blacken\path(810.380,1284.200)(687.000,1293.000)(791.723,1227.174)(810.380,1284.200)
\put(2622.000,-31.286){\arc{4901.878}{3.8024}{5.6224}}
\blacken\path(739.587,1584.958)(687.000,1473.000)(786.087,1547.039)(739.587,1584.958)
\path(3499,1338)(4480,1338)
\blacken\path(4360.000,1308.000)(4480.000,1338.000)(4360.000,1368.000)(4360.000,1308.000)
\path(2082,1338)(3089,1338)
\blacken\path(2969.000,1308.000)(3089.000,1338.000)(2969.000,1368.000)(2969.000,1308.000)
\path(642,1203)(2397,348)
\blacken\path(2275.982,373.587)(2397.000,348.000)(2302.260,427.526)(2275.982,373.587)
\path(1947,1158)(2442,438)
\blacken\path(2349.295,519.889)(2442.000,438.000)(2398.738,553.881)(2349.295,519.889)
\path(3207,1158)(2712,438)
\blacken\path(2755.262,553.881)(2712.000,438.000)(2804.705,519.889)(2755.262,553.881)
\path(4557,1203)(2757,348)
\blacken\path(2852.522,426.585)(2757.000,348.000)(2878.265,372.388)(2852.522,426.585)
\path(12,1338)(313,1338)
\blacken\path(193.000,1308.000)(313.000,1338.000)(193.000,1368.000)(193.000,1308.000)
\path(2757,258)(3747,258)
\blacken\path(3627.000,228.000)(3747.000,258.000)(3627.000,288.000)(3627.000,228.000)
\path(417,1518)(416,1521)(413,1527)
	(409,1537)(403,1551)(396,1568)
	(389,1587)(382,1607)(376,1628)
	(371,1650)(367,1673)(366,1696)
	(367,1720)(372,1743)(379,1761)
	(387,1777)(395,1789)(402,1798)
	(409,1804)(414,1809)(420,1812)
	(425,1814)(430,1816)(436,1818)
	(443,1820)(451,1823)(462,1826)
	(475,1829)(490,1832)(507,1833)
	(524,1832)(539,1829)(552,1826)
	(563,1823)(571,1820)(578,1818)
	(584,1816)(590,1814)(594,1812)
	(600,1809)(605,1804)(612,1798)
	(619,1789)(627,1777)(635,1761)
	(642,1743)(647,1720)(648,1696)
	(647,1673)(643,1650)(638,1628)
	(632,1607)(625,1587)(618,1568)
	(611,1551)(597,1518)
\blacken\path(616.249,1640.186)(597.000,1518.000)(671.483,1616.753)(616.249,1640.186)
\path(3207,1518)(3206,1521)(3203,1527)
	(3199,1537)(3193,1551)(3186,1568)
	(3179,1587)(3172,1607)(3166,1628)
	(3161,1650)(3157,1673)(3156,1696)
	(3157,1720)(3162,1743)(3169,1761)
	(3177,1777)(3185,1789)(3192,1798)
	(3199,1804)(3204,1809)(3210,1812)
	(3215,1814)(3220,1816)(3226,1818)
	(3233,1820)(3241,1823)(3252,1826)
	(3265,1829)(3280,1832)(3297,1833)
	(3314,1832)(3329,1829)(3342,1826)
	(3353,1823)(3361,1820)(3368,1818)
	(3374,1816)(3380,1814)(3384,1812)
	(3390,1809)(3395,1804)(3402,1798)
	(3409,1789)(3417,1777)(3425,1761)
	(3432,1743)(3437,1720)(3438,1696)
	(3437,1673)(3433,1650)(3428,1628)
	(3422,1607)(3415,1587)(3408,1568)
	(3401,1551)(3387,1518)
\blacken\path(3406.249,1640.186)(3387.000,1518.000)(3461.483,1616.753)(3406.249,1640.186)
\path(1767,1518)(1766,1521)(1763,1527)
	(1759,1537)(1753,1551)(1746,1568)
	(1739,1587)(1732,1607)(1726,1628)
	(1721,1650)(1717,1673)(1716,1696)
	(1717,1720)(1722,1743)(1729,1761)
	(1737,1777)(1745,1789)(1752,1798)
	(1759,1804)(1764,1809)(1770,1812)
	(1775,1814)(1780,1816)(1786,1818)
	(1793,1820)(1801,1823)(1812,1826)
	(1825,1829)(1840,1832)(1857,1833)
	(1874,1832)(1889,1829)(1902,1826)
	(1913,1823)(1921,1820)(1928,1818)
	(1934,1816)(1940,1814)(1944,1812)
	(1950,1809)(1955,1804)(1962,1798)
	(1969,1789)(1977,1777)(1985,1761)
	(1992,1743)(1997,1720)(1998,1696)
	(1997,1673)(1993,1650)(1988,1628)
	(1982,1607)(1975,1587)(1968,1568)
	(1961,1551)(1947,1518)
\blacken\path(1966.249,1640.186)(1947.000,1518.000)(2021.483,1616.753)(1966.249,1640.186)
\path(4602,1518)(4601,1521)(4598,1527)
	(4594,1537)(4588,1551)(4581,1568)
	(4574,1587)(4567,1607)(4561,1628)
	(4556,1650)(4552,1673)(4551,1696)
	(4552,1720)(4557,1743)(4564,1761)
	(4572,1777)(4580,1789)(4587,1798)
	(4594,1804)(4599,1809)(4605,1812)
	(4610,1814)(4615,1816)(4621,1818)
	(4628,1820)(4636,1823)(4647,1826)
	(4660,1829)(4675,1832)(4692,1833)
	(4709,1832)(4724,1829)(4737,1826)
	(4748,1823)(4756,1820)(4763,1818)
	(4769,1816)(4775,1814)(4779,1812)
	(4785,1809)(4790,1804)(4797,1798)
	(4804,1789)(4812,1777)(4820,1761)
	(4827,1743)(4832,1720)(4833,1696)
	(4832,1673)(4828,1650)(4823,1628)
	(4817,1607)(4810,1587)(4803,1568)
	(4796,1551)(4782,1518)
\blacken\path(4801.249,1640.186)(4782.000,1518.000)(4856.483,1616.753)(4801.249,1640.186)
\put(1857,1968){\makebox(0,0)[b]{\smash{{\SetFigFont{8}{9.6}{\familydefault}{\mddefault}{\updefault}$c,\Gamma\setminus d_2$}}}}
\put(2622,2508){\makebox(0,0)[b]{\smash{{\SetFigFont{8}{9.6}{\familydefault}{\mddefault}{\updefault}$a,c$}}}}
\put(1182,1518){\makebox(0,0)[b]{\smash{{\SetFigFont{8}{9.6}{\familydefault}{\mddefault}{\updefault}$a,b$}}}}
\put(1182,1023){\makebox(0,0)[b]{\smash{{\SetFigFont{8}{9.6}{\familydefault}{\mddefault}{\updefault}$b$}}}}
\put(2532,1428){\makebox(0,0)[b]{\smash{{\SetFigFont{8}{9.6}{\familydefault}{\mddefault}{\updefault}$a$}}}}
\put(3927,1428){\makebox(0,0)[b]{\smash{{\SetFigFont{8}{9.6}{\familydefault}{\mddefault}{\updefault}$a$}}}}
\put(3207,888){\makebox(0,0)[b]{\smash{{\SetFigFont{8}{9.6}{\familydefault}{\mddefault}{\updefault}$d_3$}}}}
\put(1317,663){\makebox(0,0)[b]{\smash{{\SetFigFont{8}{9.6}{\familydefault}{\mddefault}{\updefault}$d_1$}}}}
\put(3837,663){\makebox(0,0)[b]{\smash{{\SetFigFont{8}{9.6}{\familydefault}{\mddefault}{\updefault}$d_4$}}}}
\put(3297,1968){\makebox(0,0)[b]{\smash{{\SetFigFont{8}{9.6}{\familydefault}{\mddefault}{\updefault}$b,c,\Gamma \setminus d_3$}}}}
\put(507,1968){\makebox(0,0)[b]{\smash{{\SetFigFont{8}{9.6}{\familydefault}{\mddefault}{\updefault}$c,\Gamma\setminus d_1$}}}}
\put(4737,1968){\makebox(0,0)[b]{\smash{{\SetFigFont{8}{9.6}{\familydefault}{\mddefault}{\updefault}$b,\Gamma \setminus d_4$}}}}
\put(2307,888){\makebox(0,0)[b]{\smash{{\SetFigFont{8}{9.6}{\familydefault}{\mddefault}{\updefault}$d_2$}}}}
\put(507,1338){\ellipse{382}{382}}
\end{picture}
}
\end{center}
\caption{Quotient DFA  $\cA_6$ of prefix-free regular language with 1,296 transformations.}
\label{fig:PrFree}
\end{figure}

Any transformation $t \in T_L$ has the form 
$$t=\left( \begin{array}{cccccc}
1 & 2 & \cdots & n-2 & n-1 & n \\
i_1 & i_2 & \cdots & i_{n-2} & n & n
\end{array} \right ),
$$
where $i_k\in\{1,2,\ldots,n\}$ for $1\le k\le n-2$.
There are three cases: 
\be
\item If $i_k\le n-2$ for all $k$, $1\le k\le n-2$, then by Theorem~\ref{thm:salomaa}, $\cA_n$ can do $t$.\\
\item If $i_k\le n-1$ for all $k$, $1\le k\le n-2$, and there exists some $h$ such that $i_h= n-1$, then there exists some $j$, $1\le j\le n-2$ such that $i_k\ne j$ for all $k$, $1\le k\le n-2$.
For all $1\le k\le n-2$, define $i'_k$ as follows: $i'_k = j$ if $i_k=n-1$, and $i'_k = i_k$ if $i_k\ne n-1$. 
Let 
$$
s=\left( \begin{array}{cccccc}
1 & 2 & \cdots & n-2 & n-1 & n \\
i'_1 & i'_2 & \cdots & i'_{n-2} & n & n
\end{array} \right ).$$
By Case 1 above, $\cA_n$  can do $s$.
Since $t=sd_j$, $\cA_n$  can do $t$ as well.\\
\item Otherwise, there exists some $h$ such that $i_h= n$. Then there exists some $j$, $1\le j\le n-2$, such that $i_k\ne j$ for all $k$, $1\le k\le n-2$.
For all $1\le k\le n-2$, define  $i'_k$ as follows: $i'_k = n-1$ if $i_k=n$, $i'_k = j$ if $i_k=n-1$, and $i'_k = i_k$ otherwise.
Let $s$ be as above but with new $i'_k$.
By Case 2 above, $\cA_n$  can do $s$.
Since $t=sd_j$, $\cA_n$  can do $t$ as well.
\ee

Therefore, the syntactic complexity of $\cA_n$ meets the desired bound. \qed
\end{proof}

We conjecture that the alphabet sizes cannot be reduced. As shown in Table~\ref{tab:Summary1}, 
on p.~\pageref{table1}, 
we have verified this conjecture for $n \le 5$ by enumerating all prefix-free regular languages with $n\le 5$ using \emph{GAP}~\cite{GAP}.
\medskip

\section{Suffix-Free Regular Languages}\label{sec:sf}

For any regular language $L$, a quotient $L_w$ is \emph{uniquely reachable}~\cite{Brz09} if $L_w=L_x$ implies that $w=x$. 
It is known from~\cite{HS09} that, if $L$ is a suffix-free regular language, then $L=L_\eps$  is uniquely reachable by $\eps$, and $L$ has the empty quotient. 
Without loss of generality,  we assume that $1$ is the initial state, and $n$ is the empty state. 
We will show that the cardinality of $\Bsf(n)$, defined below, is an upper bound ($\mathbf{B}$ for ``bound'') on the syntactic complexity of suffix-free regular languages with quotient complexity $n$. Let
$$\Bsf(n) = \{ t \in \cT_{Q} \mid 1 \not\in \timg(t), \; nt = n, \txt{and for all} j\ge 1,
 \hspace{2.5cm}$$ 
$$1t^j = n \txt{or} 1t^j \neq it^j ~~\forall i, 1 < i < n\}.$$

\begin{proposition}
\label{prop:sf}
If $L$ is a regular language with quotient DFA $\cA_n = (Q, \Sig, \delta, 1, F)$ and syntactic semigroup $T_L$, then the following hold:
\be
\item If $L$ is suffix-free, then $T_L$ is a subset of $\Bsf(n)$.
\item If $L$ has the empty quotient, only one accepting quotient, and $T_L \subseteq \Bsf(n)$, then $L$ is suffix-free.
\ee
\end{proposition}

\begin{proof}
1. Let $L$ be suffix-free, and let $\cA_n$ be its quotient DFA. 
Consider an arbitrary $t \in T_L$. Since the quotient $L$ is uniquely reachable, $it \neq 1$ for all $i \in Q$. Since the quotient corresponding to state $n$ is empty, $nt = n$. 
Since $L$ is suffix-free, for any two quotients $L_w$ and $L_{uw}$, where $u,v,w \in \Sig^+$, $w = v^j$ for some $j \ge 1$, and $L_w \neq \emptyset$, we must have $L_w \cap L_{uw} = \emptyset$, and so $L_w \neq L_{uw}$. 
This means that, for any $t \in T_L$ and $j \ge 1$, if $1t^j \neq n$, then $1t^j \neq it^j$ for all $i$, $1 < i < n$. So $t \in \Bsf(n)$, and $T_L \subseteq \Bsf(n)$.

2. Assume that $T_L \subseteq \Bsf(n)$, and let $f$ be the only accepting state. If $L$ is not suffix-free, then there exist non-empty words $u$ and $v$ such that $v, uv \in L$. Let $t_u$ and $t_v$ be the transformations by $u$ and $v$, and let $i = 1t_u$; then $i \neq 1$. 
Assume without loss the generality that $n$ is the empty state. Then $f\neq n$, and we have $1t_v = f = 1t_{uv} = 1t_ut_v = it_v$, which contradicts the fact that $t_v \in \Bsf(n)$. Therefore $L$ is suffix-free. \qed
\end{proof}

Let $\bsf(n) = |\Bsf(n)|$. We now prove that $\bsf(n)$ is an upper bound  on the syntactic complexity of suffix-free regular languages. 

With each transformation $t$ of $Q$, we associate a directed graph $G_t$, where $Q$ is the set of nodes, and $(i,j) \in Q \times Q$ is a directed edge from $i$ to $j$ if $it = j$. We call such a graph $G_t$ the {\em transition graph} of $t$. For each node $i$, there is exactly one edge leaving $i$ in $G_t$. Consider the infinite sequence $i,it,it^2,\ldots$ for any $i \in Q$. Since $Q$ is finite, there exists least $j \ge 0$ such that $it^{j+1} = it^{j'}$ for some $j' \le j$. Then the finite sequence $\seq_t(i) = i,it,\ldots,it^j$ contains all the distinct elements of the above infinite sequence, and it induces a directed path $\tpath_t(i)$ from $i$ to $it^j$ in $G_t$. In particular, if $n \in \seq_t(1)$, and $nt = n$, then we call $\seq_t(1)$ the {\em principal sequence} of $t$, and $\tpath_t(1)$, the {\em principal path} of $G_t$. 
\begin{proposition}\label{prop:ppsf} 
There exists a principal sequence for every transformation~$t$ in~$\Bsf(n)$. 
\end{proposition}

\begin{proof} 
Suppose $t \in \Bsf(n)$ and $\seq_t(1) = 1,1t,\ldots,1t^j$. 
If $t$ does not have a principal sequence, 
then $n \not\in \seq_t(1)$, and $1t^{j+1} = 1t^{j'} \neq n$ for some $j' \le j$.
Let $i=1t^{j+1-j'}$; then $i \neq 1$ and $1t^{j'} =it^{j'}$, violating the last property of $\Bsf(n)$.
Therefore there is a principal sequence for every $t \in \Bsf(n)$. \qed
\end{proof}

Fix a transformation $t \in \Bsf(n)$. Let $i \in Q$ be such that $i \not\in \seq_t(1)$. If the sequence $\seq_t(i)$ does not contain any element of the principal sequence $\seq_t(1)$ other than $n$, then we say that $\seq_t(i)$ has {\em no principal connection}. Otherwise, there exists least $j \ge 1$ such that $1t^j \neq n$ and $1t^j = it^{j'} \in \seq_t(i)$ for some $j' \ge 1$, and we say that $\seq_t(i)$ has a {\em principal connection} at $1t^j$. If $j' < j$, the  principal connection is {\em short}; otherwise, it is {\em long}. 

\begin{lemma}\label{lem:shortpath} 
For all $t \in \Bsf(n)$ and $i \not\in \seq_t(1)$, the sequence $\seq_t(i)$ has no long principal connection.
\end{lemma}

\begin{proof} 
Let $t$ be any transformation in $\Bsf(n)$. Suppose for some $i \not\in \seq_t(1)$, the sequence $\seq_t(i)$ has a long principal connection at $1t^j = it^{j'} \neq n$, where $j < j'$. Hence $it^{j'-j} \neq n$, and $1t^j = (it^{j'-j})t^j$, which is a contradiction. Therefore, for all $i \not\in \seq_t(1)$, $\seq_t(i)$ has no long principal connection. \qed
\end{proof}

To calculate the cardinality of $\Bsf(n)$, we need the following observation. 
\begin{lemma}\label{lem:ptree} 

For all $t \in \Bsf(n)$ and $i \not\in \seq_t(1)$, if $\seq_t(i)$ has a principal connection, then there is no cycle incident to the path $\tpath_t(i)$ in the transition graph $G_t$. 

\end{lemma}

\begin{proof} 
This observation can be derived from Theorem 1.2.9 of~\cite{GaMa09}. However, our proof is shorter. 
Pick any $i \not\in \seq_t(1)$ such that $\seq_t(i)$ has a principal connection at $1t^j = it^{j'}$ for some $i,j$ and $j'$. Then the sequence $\seq_t(i)$ contains $n$, and the path $\tpath_t(i)$ does not contain any cycle. Suppose $C$ is a cycle which includes node 
$x=it^k \in \tpath_t(i)$. 
Since there is only one outgoing edge for each node in $G_t$, the cycle $C$ must be oriented and must contain a node $x'\not\in \tpath_t(i)$ such that $(x',x)$ is an edge in $C$.
Then the next node in the cycle must be $it^{k+1}$ since there is only one outgoing edge from $x$. But then  $x'$ can never be reached from $\tpath_t(i)$, and so no such cycle can exist. \qed
\end{proof}

By Lemma~\ref{lem:ptree}, for any $1t^j \in \seq_t(1)$, where $j \ge 1$, the union of directed paths from various nodes $i$ to $1t^j$, if $i \not\in \seq_t(1)$ and $\seq_t(i)$ has a principal connection at $1t^j$, forms a labeled tree $T_t(j)$ rooted at $it^j$. Suppose there are $r_j+1$ nodes in $T_t(j)$ for each $j$, and suppose there are $r$ elements of $Q$ that are not in the principal sequence $\seq_t(1)$ nor in any tree $T_t(j)$, for some $r_j,r \ge 0$. Note that, $it^j$ is the only node in $T_t(j)$ that is also in the principal sequence $\seq_t(1)$. Each tree $T_t(j)$ has height at most $j-1$; otherwise, some $i \in T_t(j)$ has a long principal connection. In particular, tree $T_t(1)$ has height 1; so it is trivial with only one node $1t$. Then $r_1 = 0$, and we need only consider trees $T_t(j)$ for $j \ge 2$. Let $S_m(h)$ be the number of labeled rooted trees with $m$ nodes and height at most $h$. This number can be found in the paper of Riordan~\cite{Rio60}; the calculation is somewhat complex, and we refer the reader to~\cite{Rio60} for details. For convenience, we include the values of $S_m(h)$ for small values of $m$ and $h$ in Table~\ref{tab:Smh}, where the row number is $h$ and the column number is $m$.

\begin{table}[ht]
\caption{The number $S_m(h)$ of labeled rooted trees with $m$ nodes and height at most~$h$.}
\label{tab:Smh}
\begin{center}
$
\begin{array}{|c||c|c|c|c|c|c|c|}    
\hline
 h/m & 1 & 2 & 3 & 4 & 5 & 6 & 7 \\
\hline \hline

 0   & 1 & 0 & 0 & 0 & 0 & 0 & 0 \\
\hline 

 1   & 1 & 2 & 3 & 4 & 5 & 6 & 7 \\
\hline 

 2   & 1 & 2 & 9 & 40 & 205 & 1176 & 7399 \\
\hline 

 3   & 1 & 2 & 9 & 64 & 505 & 4536 & 46249\\
\hline 

 4   & 1 & 2 & 9 & 64 & 625 & 7056 & 89929 \\
\hline 

 5   & 1 & 2 & 9 & 64 & 625 & 7776 & 112609 \\
\hline 

 6   & 1 & 2 & 9 & 64 & 625 & 7776 & 117649 \\

\hline
\end{array}
$
\end{center}
\label{table0}
\end{table}

Since each of the $m$ nodes can be the root, there are $S'_m(h) = \frac{S_m(h)}{m}$ labeled trees rooted at a fixed node and having $m$ nodes and height at most $h$. The following is an example of trees $T_t(j)$ in transformations $t \in \Bsf(n)$. 

\begin{example}\label{ex:ptree} 
Let $n = 15$. Consider any transformation $t \in \Bsf(15)$ with principal sequence $\seq_t(1) = 1,2,3,4,5,15$. There are $9$ elements of $Q$ that are not in $\seq_t(1)$, and some of them are in the trees $T_t(j)$ for $2 \le j \le 4$. Consider the cases where $r_2 = 2$, $r_3 = 3$, $r_4 = 1$, and $r = 3$. Fig.~\ref{fig:ptree} shows one such transformation $t$. 

\begin{figure}[hbt]
\begin{center}
\setlength{\unitlength}{0.00052493in}
\begingroup\makeatletter\ifx\SetFigFont\undefined%
\gdef\SetFigFont#1#2#3#4#5{%
  \reset@font\fontsize{#1}{#2pt}%
  \fontfamily{#3}\fontseries{#4}\fontshape{#5}%
  \selectfont}%
\fi\endgroup%
{\renewcommand{\dashlinestretch}{30}
\begin{picture}(6231,2507)(0,-10)
\put(462,2033){\ellipse{450}{450}}
\put(5856.375,1583.000){\arc{731.250}{5.2449}{7.3215}}
\blacken\path(6114.282,1368.376)(6042.000,1268.000)(6153.014,1322.552)(6114.282,1368.376)
\put(4962.000,2352.500){\arc{261.000}{2.3318}{7.0930}}
\blacken\path(5061.910,2381.296)(5052.000,2258.000)(5118.765,2362.127)(5061.910,2381.296)
\put(2262,2033){\ellipse{450}{450}}
\put(1362,2033){\ellipse{450}{450}}
\put(4062,2033){\ellipse{450}{450}}
\put(3162,2033){\ellipse{450}{450}}
\put(2262,1133){\ellipse{450}{450}}
\put(2712,233){\ellipse{450}{450}}
\put(3612,233){\ellipse{450}{450}}
\put(4062,1133){\ellipse{450}{450}}
\put(4962,2033){\ellipse{450}{450}}
\put(5862,233){\ellipse{450}{450}}
\put(1362,1133){\ellipse{450}{450}}
\put(3162,1133){\ellipse{450}{450}}
\put(5862,2033){\ellipse{450}{450}}
\put(5862,1133){\ellipse{450}{450}}
\path(687,2033)(1137,2033)
\blacken\path(1017.000,2003.000)(1137.000,2033.000)(1017.000,2063.000)(1017.000,2003.000)
\path(1587,2033)(2037,2033)
\blacken\path(1917.000,2003.000)(2037.000,2033.000)(1917.000,2063.000)(1917.000,2003.000)
\path(12,2033)(237,2033)
\blacken\path(117.000,2003.000)(237.000,2033.000)(117.000,2063.000)(117.000,2003.000)
\path(2487,2033)(2937,2033)
\blacken\path(2817.000,2003.000)(2937.000,2033.000)(2817.000,2063.000)(2817.000,2003.000)
\path(3387,2033)(3837,2033)
\blacken\path(3717.000,2003.000)(3837.000,2033.000)(3717.000,2063.000)(3717.000,2003.000)
\path(2262,1358)(2262,1808)
\blacken\path(2292.000,1688.000)(2262.000,1808.000)(2232.000,1688.000)(2292.000,1688.000)
\path(1497,1313)(2082,1853)
\blacken\path(2014.172,1749.562)(2082.000,1853.000)(1973.475,1793.650)(2014.172,1749.562)
\path(3162,1358)(3162,1808)
\blacken\path(3192.000,1688.000)(3162.000,1808.000)(3132.000,1688.000)(3192.000,1688.000)
\path(2847,413)(3117,908)
\blacken\path(3085.875,788.287)(3117.000,908.000)(3033.201,817.018)(3085.875,788.287)
\path(3477,413)(3207,908)
\blacken\path(3290.799,817.018)(3207.000,908.000)(3238.125,788.287)(3290.799,817.018)
\path(4062,1358)(4062,1808)
\blacken\path(4092.000,1688.000)(4062.000,1808.000)(4032.000,1688.000)(4092.000,1688.000)
\path(5862,458)(5862,908)
\blacken\path(5892.000,788.000)(5862.000,908.000)(5832.000,788.000)(5892.000,788.000)
\path(5862,1358)(5862,1808)
\blacken\path(5892.000,1688.000)(5862.000,1808.000)(5832.000,1688.000)(5892.000,1688.000)
\path(4287,2033)(4737,2033)
\blacken\path(4617.000,2003.000)(4737.000,2033.000)(4617.000,2063.000)(4617.000,2003.000)
\put(462,1975){\makebox(0,0)[b]{\smash{{\SetFigFont{7}{8.4}{\familydefault}{\mddefault}{\updefault}$1$}}}}
\put(1362,1975){\makebox(0,0)[b]{\smash{{\SetFigFont{7}{8.4}{\familydefault}{\mddefault}{\updefault}$2$}}}}
\put(2262,1975){\makebox(0,0)[b]{\smash{{\SetFigFont{7}{8.4}{\familydefault}{\mddefault}{\updefault}$3$}}}}
\put(3162,1975){\makebox(0,0)[b]{\smash{{\SetFigFont{7}{8.4}{\familydefault}{\mddefault}{\updefault}$4$}}}}
\put(4062,1975){\makebox(0,0)[b]{\smash{{\SetFigFont{7}{8.4}{\familydefault}{\mddefault}{\updefault}$5$}}}}
\put(1362,1075){\makebox(0,0)[b]{\smash{{\SetFigFont{7}{8.4}{\familydefault}{\mddefault}{\updefault}$6$}}}}
\put(2262,1075){\makebox(0,0)[b]{\smash{{\SetFigFont{7}{8.4}{\familydefault}{\mddefault}{\updefault}$7$}}}}
\put(3162,1075){\makebox(0,0)[b]{\smash{{\SetFigFont{7}{8.4}{\familydefault}{\mddefault}{\updefault}$8$}}}}
\put(2712,175){\makebox(0,0)[b]{\smash{{\SetFigFont{7}{8.4}{\familydefault}{\mddefault}{\updefault}$9$}}}}
\put(3612,175){\makebox(0,0)[b]{\smash{{\SetFigFont{7}{8.4}{\familydefault}{\mddefault}{\updefault}$10$}}}}
\put(4062,1075){\makebox(0,0)[b]{\smash{{\SetFigFont{7}{8.4}{\familydefault}{\mddefault}{\updefault}$11$}}}}
\put(5862,175){\makebox(0,0)[b]{\smash{{\SetFigFont{7}{8.4}{\familydefault}{\mddefault}{\updefault}$12$}}}}
\put(5862,1075){\makebox(0,0)[b]{\smash{{\SetFigFont{7}{8.4}{\familydefault}{\mddefault}{\updefault}$13$}}}}
\put(5862,1975){\makebox(0,0)[b]{\smash{{\SetFigFont{7}{8.4}{\familydefault}{\mddefault}{\updefault}$14$}}}}
\put(4962,1975){\makebox(0,0)[b]{\smash{{\SetFigFont{7}{8.4}{\familydefault}{\mddefault}{\updefault}$15$}}}}
\end{picture}
}
\end{center}
\caption{Transition graph of some $t \in \Bsf(15)$ with principal sequence $1,2,3,4,5,15$.}
\label{fig:ptree}
\end{figure}
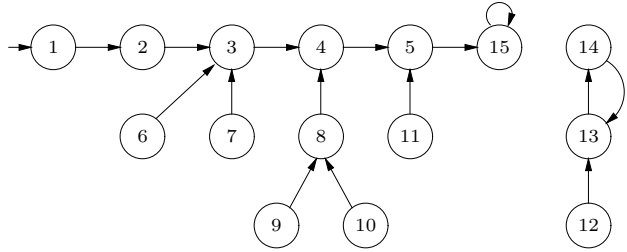

For $j = 2$, the tree $T_t(2)$ has height at most $1$, and there are $S'_{r_2+1}(1) = \frac{S_{r_2+1}(1)}{r_2+1} = \frac{3}{3} = 1$ possible $T_t(2)$. For $j=3$, there are $S'_{r_3+1}(2) = \frac{S_{r_3+1}(2)}{r_3+1} = 10$ possible $T_t(3)$, which are of one of the three types shown in Fig.~\ref{fig:Tj}. Among the 10 possible $T_t(3)$, one is of type (a), three are of type (b), and six are of type~(c). For $j = 4$, there are $S'_{r_4+1}(3) = \frac{S_{r_4+1}(3)}{r_4+1} = 1$ possible $T_t(4)$. 

\begin{figure}[hbt]
\begin{center}
\setlength{\unitlength}{0.00052493in}
\begingroup\makeatletter\ifx\SetFigFont\undefined%
\gdef\SetFigFont#1#2#3#4#5{%
  \reset@font\fontsize{#1}{#2pt}%
  \fontfamily{#3}\fontseries{#4}\fontshape{#5}%
  \selectfont}%
\fi\endgroup%
{\renewcommand{\dashlinestretch}{30}
\begin{picture}(5641,2697)(0,-10)
\put(2933,2449){\ellipse{450}{450}}
\put(908,2449){\ellipse{450}{450}}
\put(4958,2449){\ellipse{450}{450}}
\put(908,1549){\ellipse{450}{450}}
\put(1583,1549){\ellipse{450}{450}}
\put(2933,1549){\ellipse{450}{450}}
\put(2483,649){\ellipse{450}{450}}
\put(4508,1549){\ellipse{450}{450}}
\put(5408,1549){\ellipse{450}{450}}
\put(5408,649){\ellipse{450}{450}}
\put(233,1549){\ellipse{450}{450}}
\put(3383,649){\ellipse{450}{450}}
\path(2933,1774)(2933,2224)
\blacken\path(2963.000,2104.000)(2933.000,2224.000)(2903.000,2104.000)(2963.000,2104.000)
\path(2573,829)(2888,1324)
\blacken\path(2848.885,1206.654)(2888.000,1324.000)(2798.265,1238.867)(2848.885,1206.654)
\path(3293,829)(2978,1324)
\blacken\path(3067.735,1238.867)(2978.000,1324.000)(3017.115,1206.654)(3067.735,1238.867)
\path(908,1774)(908,2224)
\blacken\path(938.000,2104.000)(908.000,2224.000)(878.000,2104.000)(938.000,2104.000)
\path(368,1729)(773,2224)
\blacken\path(720.230,2112.128)(773.000,2224.000)(673.793,2150.122)(720.230,2112.128)
\path(1448,1729)(1043,2224)
\blacken\path(1142.207,2150.122)(1043.000,2224.000)(1095.770,2112.128)(1142.207,2150.122)
\path(4598,1774)(4868,2224)
\blacken\path(4831.985,2105.666)(4868.000,2224.000)(4780.536,2136.536)(4831.985,2105.666)
\path(5318,1774)(5048,2224)
\blacken\path(5135.464,2136.536)(5048.000,2224.000)(5084.015,2105.666)(5135.464,2136.536)
\path(5408,874)(5408,1324)
\blacken\path(5438.000,1204.000)(5408.000,1324.000)(5378.000,1204.000)(5438.000,1204.000)
\put(2933,2391){\makebox(0,0)[b]{\smash{{\SetFigFont{7}{8.4}{\familydefault}{\mddefault}{\updefault}$4$}}}}
\put(908,2391){\makebox(0,0)[b]{\smash{{\SetFigFont{7}{8.4}{\familydefault}{\mddefault}{\updefault}$4$}}}}
\put(4958,2391){\makebox(0,0)[b]{\smash{{\SetFigFont{7}{8.4}{\familydefault}{\mddefault}{\updefault}$4$}}}}
\put(908,64){\makebox(0,0)[b]{\smash{{\SetFigFont{7}{8.4}{\familydefault}{\mddefault}{\updefault}(a)}}}}
\put(2933,64){\makebox(0,0)[b]{\smash{{\SetFigFont{7}{8.4}{\familydefault}{\mddefault}{\updefault}(b)}}}}
\put(4958,64){\makebox(0,0)[b]{\smash{{\SetFigFont{7}{8.4}{\familydefault}{\mddefault}{\updefault}(c)}}}}
\put(233,1491){\makebox(0,0)[b]{\smash{{\SetFigFont{7}{8.4}{\familydefault}{\mddefault}{\updefault}$i_1$}}}}
\put(908,1491){\makebox(0,0)[b]{\smash{{\SetFigFont{7}{8.4}{\familydefault}{\mddefault}{\updefault}$i_2$}}}}
\put(1583,1491){\makebox(0,0)[b]{\smash{{\SetFigFont{7}{8.4}{\familydefault}{\mddefault}{\updefault}$i_3$}}}}
\put(2933,1491){\makebox(0,0)[b]{\smash{{\SetFigFont{7}{8.4}{\familydefault}{\mddefault}{\updefault}$i_1$}}}}
\put(2483,591){\makebox(0,0)[b]{\smash{{\SetFigFont{7}{8.4}{\familydefault}{\mddefault}{\updefault}$i_2$}}}}
\put(3383,591){\makebox(0,0)[b]{\smash{{\SetFigFont{7}{8.4}{\familydefault}{\mddefault}{\updefault}$i_3$}}}}
\put(4508,1491){\makebox(0,0)[b]{\smash{{\SetFigFont{7}{8.4}{\familydefault}{\mddefault}{\updefault}$i_1$}}}}
\put(5408,1491){\makebox(0,0)[b]{\smash{{\SetFigFont{7}{8.4}{\familydefault}{\mddefault}{\updefault}$i_2$}}}}
\put(5408,591){\makebox(0,0)[b]{\smash{{\SetFigFont{7}{8.4}{\familydefault}{\mddefault}{\updefault}$i_3$}}}}
\end{picture}
}
\end{center}
\caption{Three types of trees of the form $T_t(3)$, where $\{i_1,i_2,i_3\} = \{8,9,10\}$.}
\label{fig:Tj}
\end{figure}
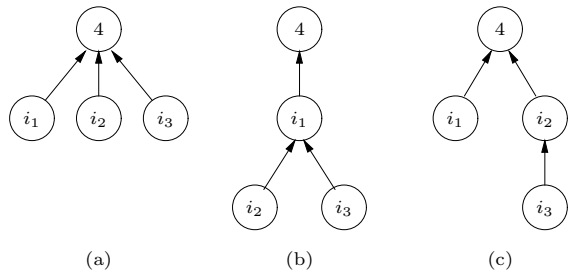
\end{example}

Let $C^n_k$ be the binomial coefficient, and let $C^n_{k_1,\ldots,k_m}$ be the multinomial coefficient. Then we have
\begin{lemma}\label{lem:Gn} 
For $n \ge 3$, we have 
\begin{equation}\label{eq:g} 
\bsf(n) = 
\sum_{k=0}^{n-2} C^{n-2}_k k! 
   \sum_{\begin{subarray}{c}
       r_2+\cdots+r_k+r\\ 
       = n-k-2
     \end{subarray}}
   C^{n-k-2}_{r_2,\ldots,r_k,r} (r+1)^r \prod_{j=2}^k S'_{r_j+1}(j-1) .
\end{equation}
\end{lemma}

\begin{proof} 
Let $t$ be any transformation in $\Bsf(n)$. Suppose $\seq_t(1) = 1,1t,\ldots,1t^k,n$ for some $k$, $0 \le k \le n-2$. There are $C^{n-2}_k k!$ different principal sequences $\seq_t(1)$. Now, fix $\seq_t(1)$. Suppose $n-k-2 = r_2 + \cdots + r_k + r$, where, for $2 \le j \le k$, tree $T_t(j)$ contains $r_j+1$ nodes, for some $r_j \ge 0$. There are $C^{n-k-2}_{r_2,\ldots,r_k,r}$ different tuples $(r_2,\ldots,r_k,r)$. Each tree $T_t(j)$ has height at most $j-1$, and it is rooted at $1t^j$. There are $S'_{r_j+1}(j-1) = \frac{S_{r_j+1}(j-1)}{r_j + 1}$ different trees $T_t(j)$. Let $E$ be the set of the remaining $r$ elements $x$ of $Q$ that are not in any tree $T_t(j)$ nor in the principal sequence $\seq_t(1)$. The image $xt$ can only be chosen from $E \cup \{n\}$. There are $(r+1)^r$ different mappings of $E$. Altogether we have the desired formula. \qed
\end{proof}

From Proposition~\ref{prop:sf} and Lemma~\ref{lem:Gn} we have 

\begin{proposition}\label{prop:Gncard} 
For $n \ge 3$, if $L$ is a suffix-free regular language with quotient complexity $n$, then its syntactic complexity $\sigma(L)$ satisfies that $\sigma(L) \le \bsf(n)$, where $\bsf(n)$ is the cardinality of $\Bsf(n)$, and it is given by Equation~(\ref{eq:g}). 
\end{proposition}

Note that $\Bsf(n)$ is not a semigroup for $n \ge 4$ because $s_1 = [2,3,n,\ldots,n,n]$, $s_2 = [n,3,3,\ldots,3,n] \in \Bsf(n)$, but $s_1s_2 = [3,3,n,\ldots,n,n] \not\in \Bsf(n)$. Hence, although $\bsf(n)$ is an upper bound on the syntactic complexity of suffix-free regular languages, that bound is not tight. Our objective is to find the largest subset of $\Bsf(n)$ that is a semigroup. Let
\begin{eqnarray*}
\Vsf(n) = \{t \in \Bsf(n) &\mid& \txt{for all} i, j \in Q \txt{where} i \neq j, \\ 
&& \txt{we have} it = jt = n \txt{or} it \neq jt\},
\end{eqnarray*}
where $\mathbf{W}$ stands for ``witness''.

\begin{proposition}\label{prop:Pncard}
For $n \ge 3$, $\Vsf(n)$ is a semigroup contained in $\Bsf(n)$, and its cardinality is
\begin{equation*}
\vsf(n) = |\Vsf(n)| = \sum_{k = 1}^{n-1} {C^{n-1}_k}(n-1-k)!{C^{n-2}_{n-1-k}}.
\end{equation*}
\end{proposition}

\begin{proof}

We know that any $t$ is in $\Vsf(n)$ if and only if the following hold: 
\be
\item $it \ne 1$ for all $i \in Q$, and $nt = n$; 
\item for all  $i,j \in Q$, such that $i\neq j$,  either $it = jt = n$ or $it \ne jt$. 
\ee

Clearly $\Vsf(n) \subseteq \Bsf(n)$. For any transformations $t_1,t_2 \in \Vsf(n)$, consider the composition $t_1t_2$. Since $1 \not\in \timg(t_2)$, we have $1 \not\in \timg(t_1t_2)$. 
We also have $nt_1t_2 = nt_2 = n$. Pick any $i,j \in Q$ such that $i\neq j$.  
Suppose $it_1t_2 \neq n$ or $jt_1t_2 \neq n$. 
If $it_1t_2 = jt_1t_2$, then $it_1 = jt_1$ and thus $i = j$, a contradiction. Hence $t_1t_2 \in \Vsf(n)$, and $\Vsf(n)$ is a semigroup contained in $\Bsf(n)$. 

Let $t \in \Vsf(n)$ be any transformation. 
Note that $nt = n$ is fixed. 
Let $Q' = Q \setminus \{n\}$, and $Q'' = Q \setminus \{1,n\}$. Suppose $k$ elements in $Q'$ are mapped to $n$ by $t$, where $0 \le k \le n-1$; then there are ${C^{n-1}_k}$ choices of these elements. For the set $D$ of the remaining $n-1-k$ elements, which must be mapped by $t$ to pairwise distinct elements of $Q''$, there are ${C^{n-2}_{n-1-k}}(n-1-k)!$ choices for the mapping $t|_D$. When $k = 0$, there is no such $t$ since $|Dt| = n-1 > n-2 = |Q''|$. 
Altogether, the cardinality of $\Vsf(n)$ is 
  $|\Vsf(n)| = \sum_{k = 1}^{n-1} {C^{n-1}_k}(n-1-k)!{C^{n-2}_{n-1-k}}. $ \qed
\end{proof}

We now construct a generating set $\Gsf(n)$ ($\mathbf{G}$ for ``generators'') of size $n$ for $\Vsf(n)$, which will show that there exist DFA's accepting suffix-free regular languages with quotient complexity $n$ and syntactic complexity $\vsf(n)$.

\begin{proposition}\label{prop:Pgen}
When $n \ge 3$, the semigroup $\Vsf(n)$ is generated by the following set $\Gsf(n)$ of transformations of $Q$:
$\Gsf(3)=\{a,b\}$, where $a=[3,2,3]$ and $b=[2,3,3]$; $\Gsf(4)=\{a,b,c\}$, where $a= [4, 3, 2, 4]$, $b = [2, 4, 3, 4]$, $c = [2, 3, 4, 4]$; and for $n\ge 5$, $\Gsf(n)=\{a_0,\ldots,a_{n-1}\}$, where

\bi
\item $a_0={1 \choose n}(2,3)$,
\item $a_1={1 \choose n}(2,3,\ldots,n-1)$,
\item For $2 \le i \le n - 1$, $ja_i = j+1$ for $j = 1,\ldots,i-1$, $ia_i = n$, and $ja_i = j$ for $j = i+1,\ldots,n$.
\ei

\end{proposition}

\begin{proof}
First note that $\Gsf(n)$ is a subset of $\Vsf(n)$, and so $\langle \Gsf(n) \rangle$, the semigroup generated by $\Gsf(n)$, is a subset of $\Vsf(n)$. We now show that $\Vsf(n) \subseteq \langle \Gsf(n) \rangle$. 

Pick any $t$ in $\Vsf(n)$. Note that $nt = n$ is fixed. 
Let $Q' = Q \setminus \{n\}$, $E_t = \{j \in Q' \mid jt = n\}$,  $D_t = Q' \setminus E_t$, and $Q'' = Q \setminus \{1,n\}$. Then $D_t t \subseteq Q''$, and $|E_t| \ge 1$, since $|Q''| < |Q'|$.  We prove by induction on $|E_t|$ that $t \in \langle \Gsf(n) \rangle$. 

First, note that $\langle a_0,a_1 \rangle$, the semigroup generated by $\{a_0,a_1\}$, is isomorphic to the symmetric group $\mathfrak{S}_{Q''}$ by Theorem~\ref{thm:piccard}. Consider $E_t = \{i\}$ for some $i \in Q'$. Then $ia_i = it = n$. Moreover, since $D_t a_i, D_t t \subseteq Q''$, there exists $\pi \in \langle a_0, a_1 \rangle$ such that $(ja_i)\pi = jt$ for all $j \in D_t$. Then $t = a_i\pi \in \langle \Gsf(n) \rangle$.

Assume that any transformation $t \in \Vsf(n)$ with $|E_t| < k$ can be generated by $\Gsf(n)$, where $1 < k < n-1$. 
Consider $t \in \Vsf(n)$ with $|E_t| = k$. 
Suppose $E_t = \{e_1,\ldots,e_{k-1},e_k\}$. 
Let $s\in \Vsf(n)$ be such that $E_s = \{e_1,\ldots,e_{k-1}\}$. By assumption, $s$ can be generated by $\Gsf(n)$. 
Let $i = e_ks$; 
then $i \in Q''$, and $e_j(s a_i) = n$ for all $1 \le j \le k$. 
Moreover, we have $D_t (s a_i) \subseteq Q''$.
Thus, there exists $\pi \in \langle a_0, a_1 \rangle$ such that, 
for all $d \in D_t$, $d(s a_i \pi) = d t$. 
Altogether, for all $e_j \in E_t$, we have $e_j (s a_i \pi) = e_j t = n$, for all $d \in D_t$, $d (s a_i \pi) = d t$, and $n (s a_i \pi) = n t = n$. Thus $t = s a_i \pi$, and $t \in \langle \Gsf(n) \rangle$.

Therefore $\Vsf(n) = \langle \Gsf(n) \rangle$. \qed
\end{proof}

\begin{theorem}
\label{thm:DFAsf}
For $n \ge 5$, let $\cA_n = (Q,\Sig,\delta,1,F)$ be the DFA with alphabet $\Sig = \{a_0,a_1,\ldots,a_{n-1}\}$, where each $a_i$ defines a transformation as in Proposition~\ref{prop:Pgen}, and $F = \{2\}$. 
Then $L = L(\cA_n)$ has quotient complexity $\kappa(L) = n$, and syntactic complexity $\sigma(L) = \vsf(n)$. Moreover, $L$ is suffix-free.
\end{theorem}

\begin{proof}
First we show that all the states of $\cA_n$ are reachable: $1$ is the initial state,  state $n$ is reached by $a_1$, and for $2 \leq i \leq n-1$, state $i$ is reached by $a_i^{i-1}$. 
Also, the initial state $1$ accepts $a_2$ while state $i$ rejects $a_2$ for all $i \neq 1$. 
For $2 \leq i \leq n-1$,  state $i$ accepts $ a_1^{n-i}$, while state $j$ rejects it, for all $j \neq i$. Also $n$ is the empty state. Thus all the states of $\cA_n$ are distinct, and $\kappa(L) = n$.

By Proposition~\ref{prop:Pgen}, the syntactic semigroup of $L$ is $\Vsf(n)$. The syntactic complexity of $L$ is $\sigma(L) = |\Vsf(n)| = \vsf(n)$. Also, by Proposition~\ref{prop:sf}, $L$ is suffix-free. \qed
\end{proof}

As shown in Table~\ref{tab:Summary1} on p.~\pageref{table1}, the size of $\Sig$ cannot be decreased for $n\le 5$.

\begin{theorem}\label{thm:sfsmall} 
For $2 \le n \le 5$, if a suffix-free regular language $L$ has quotient complexity $\kappa(L) = n$, then its syntactic complexity satisfies that $\sigma(L) \le \vsf(n)$, and this is a tight upper bound. 
\end{theorem}

\begin{proof} 
By Proposition~\ref{prop:sf}, the syntactic semigroup of a suffix-free regular language $L$ is contained in $\Bsf(n)$. 
For $n\in\{2,3\}$, $\vsf(n)=\bsf(n)$. 
So $\vsf(n)$ is an upper bound, and it is met by the language $L = \eps$ for $n = 2$ and by $L = ab^*$ for $n = 3$. 
For $n=4$, we have $|\Bsf(4)| = 15$ and $|\Vsf(4)| = 13$. 
Two transformations, $s_1 = [4, 2, 2, 4]$ and $s_2 = [4, 3, 3, 4]$, in $\Bsf(4)$ are such that $s_1$ conflicts with $t_1 = [3, 2, 4, 4] \in \Vsf(4)$ ($t_1s_1 = [2,2,4,4] \not\in \Bsf(4)$), and $s_2$ conflicts with $t_2 = [2,3,4,4]$ ($t_2s_2 = [3,3,4,4] \not\in \Bsf(4)$). 
Thus $\sigma(L) \le 13$. 
Let $L = (b \cup c)((a \cup c)b^*a)^*$; then $\kappa(L) = 4$ and $\sigma(L) = 13$. So the bound is tight.

For $n=5$, we have $|\Bsf(5)| = 115$ and $|\Vsf(5)| = 73$. Let $\Bsf(5) \setminus \Vsf(5) = \{s_1,\ldots,s_{42}\}$. For each $s_i$, we enumerated transformations in $\Vsf(5)$ using \emph{GAP} and found a unique $t_i \in \Vsf(5)$ such that the semigroup $\langle t_i, s_i \rangle$ is not contained in $\Bsf(5)$. Thus at most one transformation in each pair $\{t_i,s_i\}$ can appear in the syntactic semigroup of $L$. So we reduce the upper bound to~$73$. By Theorem~\ref{thm:DFAsf}, this bound is tight. 
\end{proof}
\smallskip

For $n \ge 6$, the semigroup $\Vsf(n)$ is no longer the largest semigroup contained in $\Bsf(n)$. In the following, we define and study another semigroup $\Wsf(n)$, which is a larger semigroup contained in $\Bsf(n)$. Let 
$$\Wsf(n) = \{ t \in \Bsf(n) \mid 1t = n \txt{or} it = n~~\forall~i, 2 \le i \le n-1 \}.$$ 
Note that, we are interested only in situations where $n \ge 6$, although some statements also hold for smaller $n$.

\begin{proposition}\label{prop:Wsf} 
For $n \ge 6$, the set $\Wsf(n)$ is a semigroup contained in $\Bsf(n)$, and its cardinality is 
$$\wsf(n) = |\Wsf(n)| = (n-1)^{n-2} + (n-2).$$
\end{proposition}

\begin{proof} 
Pick any $t_1,t_2$ in $\Wsf(n)$. If $1t_1 = n$, then $1(t_1t_2) = n$ and $t_1t_2 \in \Wsf(n)$. If $1t_1 \neq n$, then, for all $i \in \{2,\ldots,n-1\}$, $it_1 = n$ and $i(t_1t_2) = n$; so $t_1t_2 \in \Wsf(n)$ as well. Hence $\Wsf(n)$ is a semigroup contained in $\Bsf(n)$. 

For any $t \in \Wsf(n)$, $nt = n$ is fixed. There are two possible cases: 
\be
\item $1t = n$: For each $i \in \{2,\ldots,n-1\}$, $it$ can be chosen from $\{2,\ldots,n\}$. Then there are $(n-1)^{n-2}$ different $t$'s in this case. 
\item $1t \neq n$: Now $1t$ can be chosen from $\{2,\ldots,n-1\}$. For each $i \in \{2,\ldots,n-1\}$, $it = n$ is fixed. There are $n-2$ different $t$'s in this case. 
\ee
Therefore $\wsf(n) = (n-1)^{n-2} + (n-2)$. \qed
\end{proof}


\begin{proposition}\label{prop:Wsfgen} 
For $n \ge 6$, the semigroup $\Wsf(n)$ is generated by the set $\Hsf(n) = \{a_1,a_2,a_3,b_1,\ldots,b_{n-2},c\}$ of transformations, where 
\be 
\item $a_1 = {1 \choose n}(2,\ldots,n-1)$, $a_2 = {1 \choose n}(2,3)$, $a_3 = {1 \choose n}{n-1 \choose 2}$; 
\item For $1 \le i \le n-2$, $b_i = {1 \choose n}{i+1 \choose n}$; 
\item $c = {Q \setminus \{1\} \choose n}{1 \choose 2} = [2,n,\ldots,n]$. 
\ee 
\end{proposition}

\begin{proof} 
Clearly $\Hsf(n) \subseteq \Wsf(n)$, and $\langle \Hsf(n) \rangle \subseteq \Wsf(n)$. We show in the following that $\Wsf(n) \subseteq \langle \Hsf(n) \rangle$. 

Let $Q' = \{2,\ldots,n-1\}$. By Theorem~\ref{thm:salomaa}, $a_1,a_2$ and $a_3$ together generate the semigroup $$\mathbf{Y} = \{t \in \Wsf(n) \mid \txt{for all} i \in Q', it \in Q'\},$$ which is isomorphic to $\cT_{Q'}$ and is contained in $\Wsf(n)$. Next, consider any $t \in \Wsf(n) \setminus \mathbf{Y}$. We have two cases: 
\be 
\item $1t = n$: Let $E_t = \{ i \in Q' \mid it = n \}$. Since $t \not\in \mathbf{Y}$, $E_t \neq \emptyset$. Suppose $E_t = \{i_1,\ldots,i_k\}$, for some $1 \le k \le n-2$. Then there exists $t' \in \mathbf{Y}$ such that, for all $i \not\in E_t$, $it' = it$. Let $s = b_{i_1-1} \cdots b_{i_k-1}$. Note that $E_ts = \{n\}$, and, for all $i \not\in E_t$, $i(t's) = (it')s = it$. So $t = t's \in \langle \Hsf(n) \rangle$. 
\item $1t \neq n$: If $1t = 2$, then $t = c$. Otherwise, $1t \in \{3,\ldots,n-1\} \subseteq Q'$, and we know from the above case that there exists $t' \in \Hsf(n)$ such that $2t' = 1t$. Then $1(ct') = 1t$, and $i(ct') = (ic)t' = n = it$, for all $i \in Q'$. Hence $t = ct'~\in~\langle \Hsf(n) \rangle$. 
\ee 

Therefore $\langle a_1,a_2,a_3,b_1,\ldots,b_{n-2},c \rangle = \Wsf(n)$. \qed
\end{proof}

\begin{theorem}\label{thm:wsfaut} 
For $n \ge 6$, let $\cA'_n = (Q, \Sig, \delta, 1, F)$ be the DFA with alphabet $\Sig = \{a_1,a_2,a_3,b_1,\ldots,b_{n-2},c\}$ of size $n+2$, where each letter defines a transformation as in Proposition~\ref{prop:Wsfgen}, and $F = \{2\}$. Then $L' = L(\cA'_n)$ has quotient complexity $\kappa(L') = n$ and syntactic complexity $\sigma(L') = \wsf(n)$. 
\end{theorem}

\begin{proof}
First we show that $\kappa(L') = n$. From the initial state, we can reach state $2$ by $c$ and state $n$ by $a_1$. From state $2$ we can reach state $i$, $3 \le i \le n-1$, by $a_1^{i-1}$. So all the states in $Q$ are reachable. Now, the initial state accepts $c$, but all other states reject it. For $2 \le i \le n-2$, state $i$ accepts $a_1^{n-i}$, while all other states reject it. State $n$ is the empty state, which rejects all words. Thus all the states in $Q$ are distinct. 

By Proposition~\ref{prop:Wsfgen}, the syntactic semigroup of $L'$ is $\Wsf(n)$, and $\sigma(L') = \wsf(n)$. Also $L'$ is suffix-free by Proposition~\ref{prop:sf}. \qed
\end{proof}

We know that the upper bound on the  syntactic complexity of suffix-free regular languages is achieved by the largest semigroup contained in $\Bsf(n)$. We conjecture that $\Wsf(n)$ is such a semigroup. 
\medskip

\begin{conjecture}[Suffix-Free Regular Languages]
\label{con:sf}
If $L$ is a suffix-free regular language with $\kappa(L) = n \ge 6$, then $\sigma(L) \le \wsf(n)$ and this is a tight bound. 
\end{conjecture}


We prove the conjecture for $n = 6$:

\begin{proof} 
For $n = 6$,  $|\Bsf(6)| = 1169$ and $|\Wsf(6)| = 629$. Let $\{s_1,\ldots,s_{540}\} = \Bsf(6) \setminus \Wsf(6)$. For each $i$, we enumerated transformations in $\Wsf(6)$ using \emph{GAP} and found a unique $t_i \in \Wsf(6)$ such that $\langle t_i, s_i \rangle$ is not contained in $\Bsf(6)$. As in the proof of Theorem~\ref{thm:sfsmall}, for each $i$, at most one transformation in $\{t_i,s_i\}$ can appear in the syntactic semigroup of $L$. Then we can reduce the upper bound to $629$. This bound is met by the language $L'$ in Theorem~\ref{thm:wsfaut}; so it is tight. \qed
\end{proof}

\section{Bifix-Free Regular Languages}\label{sec:bf}

Let $L$ be a regular bifix-free language with $\kappa(L) = n$. From Sections~\ref{sec:pf} and~\ref{sec:sf} we have: 
\be
\item $L$ has $\eps$ as a quotient, and this is the only accepting quotient; 
\item $L$ has $\emptyset$ as a quotient; 
\item $L$ as a quotient is uniquely reachable.
\ee

Let $\cA$ be the quotient DFA of $L$, with $Q$ as the set of states. We assume that $1$ is the initial state, $n-1$ corresponds to the quotient $\eps$, and $n$ is the empty state. Consider the set
$$\Bbf(n) = \{ t \in \Bsf(n) \mid (n-1) t = n \}.$$
The following is an observation similar to Proposition~\ref{prop:sf}.

\begin{proposition}\label{prop:bf}
If $L$ is a regular language with quotient complexity $n$ and syntactic semigroup $T_L$, then the following hold: 
\be
\item If $L$ is bifix-free, then $T_L$ is a subset of $\Bbf(n)$. 
\item If $\eps$ is the only accepting quotient of $L$, and $T_L \subseteq \Bbf(n)$, then $L$ is bifix-free.
\ee
\end{proposition}

\begin{proof}\mbox{}
1. Since $L$ is suffix-free,  $T_L \subseteq \Bsf(n)$. Since $L$ is also prefix-free,  it has $\eps$ and $\emp$ as quotients. By assumption, $n-1 \in Q$ corresponds to the quotient $\eps$. Thus for any $t \in T_L$, $(n-1) t = n$, and so $T_L \subseteq \Bbf(n)$. 

2. Since $\eps$ is the only accepting quotient of $L$, $L$ is prefix-free, and $L$ has the empty quotient. Since $T_L \subseteq \Bbf(n) \subseteq \Bsf(n)$,  $L$ is suffix-free by Proposition~\ref{prop:sf}. Therefore $L$ is bifix-free. \qed
\end{proof}

\begin{lemma}\label{lem:Hn} 
For $n \ge 3$, we have $|\Bbf(n)| = M_n + N_n$, where 
\begin{eqnarray} 
\label{eq:H1} 
  M_n &=& \sum_{k=1}^{n-2} C^{n-3}_{k-1} (k-1)! \sum_{\begin{subarray}{c} 
      r_2 + \cdots + r_k + r \\ 
          = n-k-2
    \end{subarray}}
    C^{n-k-2}_{r_2,\ldots,r_k,r} (r+1)^r \prod_{j=2}^k S'_{r_j+1}(j-1), \\
\label{eq:H2} 
  N_n &=& \sum_{k=0}^{n-3} C^{n-3}_k k! \sum_{\begin{subarray}{c} 
      r_2 + \cdots + r_k + r \\ 
          = n-k-3
    \end{subarray}}
    C^{n-k-3}_{r_2,\ldots,r_k,r} (r+2)^r \prod_{j=2}^k S'_{r_j+1}(j-1) .
\end{eqnarray}
\end{lemma}

\begin{proof} 
Let $t$ be any transformation in $\Bbf(n)$. Suppose $\seq_t(1) = 1,1t,\ldots,1t^k,n$, where $0 \le k \le n-2$. For $2 \le j \le k$, suppose tree $T_t(j)$ contains $r_j+1$ nodes, for some $r_j \ge 0$; then there are $S'_{r_j+1}(j-1)$ different trees $T_t(j)$. Let $E$ be the set of elements of $Q$ that are not in any tree $T_t(j)$ nor in the principal sequence~$\seq_t(1)$. Then there are two cases: 

\be 

\item $n-1 \in \seq_t(1)$: Since $(n-1)t = n$, we must have $1t^k = n-1$, and $k \ge 1$. So there are $C^{n-3}_{k-1} (k-1)!$ different $\seq_t(1)$. Let $r = |E| = (n-k-2) - (r_2 + \cdots + r_k)$. Then there are $C^{n-k-2}_{r_2,\ldots,r_k,r}$ tuples $(r_2,\ldots,r_k,r)$. For any $x \in E$, its image $xt$ can be chosen from $E \cup \{n\}$. Then the number of transformations $t$ in this case is $M_n$. 

\item $n-1 \not\in \seq_t(1)$: Then $k \le n-3$, and there are $C^{n-3}_k k!$ different $\seq_t(1)$. Note that $n-1 \in E$, and $(n-1)t = n$ is fixed. Let $r = |E \setminus \{n-1\}| = (n-k-3) - (r_2 + \cdots + r_k)$. Then there are $C^{n-k-3}_{r_2,\ldots,r_k,r}$ tuples $(r_2,\ldots,r_k,r)$. For any $x \in E \setminus \{n-1\}$, $xt$ can be chosen from $E \cup \{n\}$. Thus the number of transformations $t$ in this case is $N_n$. 

\ee

Altogether we have the desired formula. \qed
\end{proof}

Let $\bbf(n) = |\Bbf(n)|$. From Proposition~\ref{prop:bf} and Lemma~\ref{lem:Hn} we have 
\begin{proposition}\label{prop:Hncard} 
For $n \ge 3$, if $L$ is a bifix-free regular language with quotient complexity $n$, then its syntactic complexity $\sigma(L)$ satisfies that $\sigma(L) \le \bbf(n)$, where $\bbf(n)$ is the cardinality of $\Bbf(n)$ as in Lemma~\ref{lem:Hn}.
\end{proposition}

For $2 \le n \le 4$, the set $\Bbf(n)$ is a semigroup. But for $n \ge 5$, it is not a semigroup because $s_1 = [2,3,n,\ldots,n,n]$, $s_2 = [n,3,3,n,\ldots,n,n] \in \Bbf(n)$ while $s_1s_2 = [3,3,n,\ldots,n,n] \not\in \Bbf(n)$. Hence $\bbf(n)$ is not a tight upper bound on the syntactic complexity of bifix-free regular languages in general. We look for a large semigroup contained in $\Bbf(n)$ that can be the syntactic semigroup of a bifix-free regular language. Let 
\begin{eqnarray*}
\Vbf(n) = \{ t \in \Bbf(n) &\mid& \txt{for all} i, j \in Q \txt{where} i \neq j, \\
&& \txt{we have} it = jt = n \txt{or} it \neq jt \}.
\end{eqnarray*}
(The reason for using the superscript $\le 5$ will be made clear in Theorem~\ref{thm:bfsmall}.)

\begin{proposition}\label{prop:Vncard}
For $n \ge 3$, $\Vbf(n)$ is a semigroup contained in $\Bbf(n)$ with cardinality
\begin{equation*}
\vbf(n) = |\Vbf(n)| = \sum_{k=0}^{n-2} \big( C^{n-2}_k \big)^2 (n-2-k)!
\end{equation*}
\end{proposition}

\begin{proof}
First, note that $\Vbf(n) = \Vsf(n) \cap \Bbf(n)$, and that $\Vsf(n)$ is a semigroup contained in $\Bsf(n)$ by Proposition~\ref{prop:Pncard}. For any $t_1,t_2 \in \Vbf(n)$, we have $t_1t_2 \in \Vsf(n)$, and $(n-1)t_1t_2 = nt_2 = n$; so $t_1t_2 \in \Bbf(n)$. Then $t_1t_2 \in \Vbf(n)$, and $\Vbf(n)$ is a semigroup contained in~$\Bbf(n)$. 

Pick any $t \in \Vbf(n)$. Note that $(n-1)t = n$ and $nt = n$ are fixed, and $1 \not\in \timg(t)$. Let $Q' = Q \setminus \{n-1,n\}$, $E = \{i \in Q' \mid it = n\}$, and $D = Q' \setminus E$. Suppose $|E| = k$, where $0 \leq k \leq n-2$; then there are ${C^{n-2}_k}$ choices of $E$. Elements of $D$ are mapped to pairwise different elements of $Q \setminus \{1,n\}$; then there are ${C^{n-2}_{n-2-k}}(n-2-k)!$ different mappings $t|_D$. Altogether, we have
$  |\Vbf(n)| = \sum_{k=0}^{n-2} \big( C^{n-2}_k \big)^2 (n-2-k)! $ \qed
\end{proof}

\begin{proposition}\label{prop:bfgen}
For $n \ge 3$, let $Q' = Q \setminus \{n-1, n\}$ and $Q'' = Q \setminus \{1,n\}$. Then the semigroup $\Vbf(n)$ is 
generated by
$$\Gbf(n) = \{t \in \Vbf(n) \mid Q't = Q'' \txt{and} it \neq jt \txt{for all} i,j \in Q'\}.$$
\end{proposition}

\begin{proof}\label{proof:bfgen}
We want to show that $\Vbf(n) = \langle \Gbf(n) \rangle$. Since $\Gbf(n) \subseteq \Vbf(n)$, we have $\langle \Gbf(n) \rangle \subseteq \Vbf(n)$. 
Let $t \in \Vbf(n)$. By definition, $(n-1)t = nt = n$. Let $E_{t} = \{i \in Q' \mid it = n\}$. If $E_{t} = \emptyset$, then $t \in \Gbf(n)$; 
otherwise,  there exists $x \in Q''$ such that $x \not\in \timg(t)$. We prove by induction on $|E_{t}|$ that $t \in \langle \Gbf(n) \rangle$.

First note that, for all $t \in \Gbf(n)$, $t|_{Q'}$ is an injective mapping from $Q'$ to $Q''$. Consider  $E_{t} = \{i\}$ for some $i \in Q'$. 
Since $|E_{t}| = 1$, $\timg(t) \cup \{x\} = Q''$. Let $t_1,t_2 \in \Gbf(n)$ be defined by
\be
\item $jt_1 = j+1$ for $j = 1,\ldots,i-1$, $it_1 = n-1$, $jt_1 = j$ for $j = i+1,\ldots,n-2$, 
\item $1t_2 = x$, $jt_2 = (j-1)t$ for $j = 2,\ldots,i$, $jt_2 = jt$ for $j = i+1,\ldots,n-2$. 
\ee
Then $t_1t_2 = t$, and $t \in \langle \Gbf(n) \rangle$.
\goodbreak

Assume that any transformation $t \in \Vbf(n)$ with $|E_{t}| < k$ can be generated by $\Gbf(n)$, where $1 < k < n-2$. 
Consider $t \in \Vbf(n)$ with $|E_{t}| = k$. 
Suppose $E_{t} = \{e_1,\ldots,e_{k-1},e_k\}$, and let $D_{t} = Q' \setminus E_{t} = \{d_1,\ldots,d_l\}$, where $l = n - 2 - k$. 
By assumption, all $s \in \Vbf(n)$ with $|E_{s}| = k-1$ can be generated by $\Gbf(n)$. 
Let $s$ be such that $E_{s} = \{1,\ldots,k-1\}$; then $1s = \cdots = (k-1)s = n$. In addition, let $ks = x$, and let $(k+j)s = d_jt$ for $j = 1,\ldots,l$. Let $t' \in \Gbf(n)$ be such that $e_jt' = j$ for $j = 1,\ldots,k-1$, $kt' = n-1$, and $d_jt' = k+j$ for $j = 1,\ldots,l$. Then $t's = t$, and $t \in \langle \Gbf(n) \rangle$.
Therefore, $\Vbf(n) = \langle \Gbf(n) \rangle$. \qed
\end{proof}

\begin{theorem}\label{thm:bfaut}
For $n \ge 3$, let $\cA_n = (Q,\Sig,\delta,1,F)$ be the DFA 
with alphabet $\Sig$ of size $(n-2)!$, where each $a \in \Sig$ defines a distinct transformation $t_a \in \Gbf(n)$, and $F = \{n-1\}$. 
Then $L = L(\cA_n)$ has quotient complexity $\kappa(L) = n$, and syntactic complexity $\sigma(L) = \vbf(n)$. 
Moreover, $L$ is bifix-free.
\end{theorem}

\begin{proof}\label{proof:bfaut}
We first show that all the states of $\cA_n$ are reachable. 
Note that there exists $a \in \Sig$ such that $t_a = [2,3,\ldots,n-1,n,n] \in \Gbf(n)$. 
State $1 \in Q$ is the initial state, and
 $a^{i-1}$ reaches  state $i \in Q$ for $i = 2,\ldots,n$. 
Furthermore, for $1 \leq i \leq n-1$, state $i$ accepts  $ a^{n-1-j}$, while for $j \neq i$,  state $j$ rejects it. Also, $n$ is the empty state. Thus all the states of $\cA_n$ are distinct, and $\kappa(L) = n$. 

By Proposition~\ref{prop:bfgen}, the syntactic semigroup of $L$ is $\Vbf(n)$. Hence the syntactic complexity of $L$ is $\sigma(L) = \vbf(n)$. By Proposition~\ref{prop:bf}, $L$ is bifix-free. \qed
\end{proof}

\begin{theorem}\label{thm:bfsmall} 
For $2 \le n \le 5$, if a bifix-free regular language $L$ has quotient complexity $\kappa(L) = n$, then $\sigma(L) \le \vbf(n)$, and this bound is tight. 
\end{theorem}

\begin{proof}
We know by Proposition~\ref{prop:bf} that the upper bound on the syntactic complexity of bifix-free regular languages is reached by the largest semigroup contained in $\Bbf(n)$. Since $\vbf(n) = \bbf(n)$ for $n = 2$, $3$, and $4$, $\vbf(n)$ is an upper bound, and it is tight by Theorem~\ref{thm:bfaut}. 

For $n = 5$, we have $\bbf(5) = |\Bbf(5)| = 41$, and $\vbf(5) = |\Vbf(5)| = 34$. Let $\Bbf(5) \setminus \Vbf(5) = \{\tau_1,\ldots,\tau_7\}$. We found for each $\tau_i$ a unique $t_i \in \Vbf(5)$ such that the semigroup $\langle \tau_i, t_i \rangle$ is not a subset of $\Bbf(5)$:
$$\begin{array}{ll}
 \tau_1 = [ 2, 4, 4, 5, 5 ], \quad & \quad t_1 = [ 3, 4, 2, 5, 5 ]; \\
 \tau_2 = [ 3, 4, 4, 5, 5 ], \quad & \quad t_2 = [ 3, 5, 2, 5, 5 ]; \\
 \tau_3 = [ 4, 2, 2, 5, 5 ], \quad & \quad t_3 = [ 2, 4, 3, 5, 5 ]; \\
 \tau_4 = [ 4, 3, 3, 5, 5 ], \quad & \quad t_4 = [ 2, 5, 3, 5, 5 ]; \\
 \tau_5 = [ 5, 2, 2, 5, 5 ], \quad & \quad t_5 = [ 3, 2, 4, 5, 5 ]; \\
 \tau_6 = [ 5, 3, 3, 5, 5 ], \quad & \quad t_6 = [ 2, 3, 4, 5, 5 ]; \\
 \tau_7 = [ 5, 4, 4, 5, 5 ], \quad & \quad t_7 = [ 3, 2, 5, 5, 5 ].
\end{array}$$

Since $\langle \tau_i, t_i \rangle \subseteq T_L$, if both $\tau_i$ and $t_i$ are in $T_L$,
then $T_L \not\subseteq \Bbf(5)$, and  $L$ is not bifix-free by Proposition~\ref{prop:bf}. Thus, for $1 \le i \le 7$, at most one of $\tau_i$ and $t_i$ can appear in $T_L$, and $|T_L| \le 34$. Since $|\Vbf(5)| = 34$ and $\Vbf(5)$ is a semigroup, we have $\sigma(L) \le 34 = \vbf(5)$ as the upper bound for $n = 5$. This bound is reached by the DFA $\cA_5$ in Theorem~\ref{thm:bfaut}. \qed 
\end{proof}

For $n \ge 6$, the semigroup $\Vbf(n)$ is no longer the largest semigroup contained in $\Bbf(n)$. We find another large semigroup $\Wbf(n)$ suitable for bifix-free regular languages. Let 
\begin{eqnarray*}
  \Uf^1_n &=& \{t \in \Bbf(n) \mid 1t = n \}, \\ 
  \Uf^2_n &=& \{t \in \Bbf(n) \mid 1t = n-1\}, \\
  \Uf^3_n &=& \{t \in \Bbf(n) \mid 1t \not\in \{n, n-1\}, \txt{and} it \in \{n-1, n\} \txt{for all} i \neq 1\}, 
\end{eqnarray*}
and let $\Wbf(n) = \Uf^1_n \cup \Uf^2_n \cup \Uf^3_n$. 
When $2 \le n \le 4$, we have $\Wbf(n)~=~\Vbf(n)$, and these cases were already discussed. So we are only interested in larger values~of~$n$. 

\begin{proposition}\label{prop:Rcard} 
For $n \ge 5$, $\Wbf(n)$ is a semigroup contained in $\Bbf(n)$ with cardinality 
$$ \wbf(n) = |\Wbf(n)| = (n-1)^{n-3} + (n-2)^{n-3} + (n-3)2^{n-3}. $$
\end{proposition}

\begin{proof} 

First we show that $\Uf^1_n$ is a semigroup. For any $t_1,t_1' \in \Uf^1_n$, since $1(t_1t_1') = (1t_1)t_1' = nt_1' = n$, we have $t_1t_1' \in \Uf^1_n$. Next, let $t_2 \in \Uf^2_n$ and $t \in \Uf^1_n \cup \Uf^2_n$. If $t \in \Uf^1_n$, then $1(t_2t) = (n-1)t = n$ and $1(tt_2) = nt_2 = n$; so $t_2t, tt_2 \in \Uf^1_n$. If $t \in \Uf^2_n$, then $1(t_2t) = (n-1)t = n$ and $1(tt_2) = (n-1)t_2 = n$; so $t_2t, tt_2 \in \Uf^1_n$ as well. Thus $\Uf^1_n \cup \Uf^2_n$ is also a semigroup. For any $t_3 \in \Uf^3_n$ and $t' \in \Wbf(n)$, since $it_3 \in \{n-1,n\}$ for all $i \neq 1$, and $(n-1)t' = nt' = n$, we have $i(t_3t') = n$, and $t_3t' \in \Wbf(n)$. Also $1(t't_3) = (1t')t_3 \in \{n-1,n\}$, so $t't_3 \in \Uf^1_n \cup \Uf^2_n$. Hence $\Wbf(n)$ is a semigroup contained in $\Bbf(n)$.

Note that $\Uf^1_n$, $\Uf^2_n$, and $\Uf^3_n$ are pairwise disjoint. For any $t \in \Wbf(n)$, there are three cases: 

\be
\item $t \in \Uf^1_n$: For any $i \not\in \{1,n-1,n\}$, $it$ can be chosen from $Q \setminus \{1\}$. Then $|\Uf^1_n| = (n-1)^{n-3}$; 
\item $t \in \Uf^2_n$: For any $i \not\in \{1,n-1,n\}$, $it$ can be chosen from $Q \setminus \{1,n-1\}$. Then $|\Uf^2_n| = (n-2)^{n-3}$; 
\item $t \in \Uf^3_n$: Now, $1t$ can be chosen from $Q \setminus \{1,n-1,n\}$. For any $i \not\in \{1,n-1,n\}$, $it$ has two choices: $it = n-1$ or $n$. Then $|\Uf^3_n| = (n-3)2^{n-3}$. 
\ee 
Therefore we have $|\Wbf(n)| = (n-1)^{n-3} + (n-2)^{n-3} + (n-3)2^{n-3}$. \qed
\end{proof}

The next proposition describes a generating set of $\Wbf(n)$. 
\begin{proposition}\label{prop:Rgen} 
For $n \ge 5$, the semigroup $\Wbf(n)$ is generated by $\Hbf(n) = \{a_1, a_2, a_3, b_1, \ldots, b_{n-3}, c_1, \ldots, c_m, d_1, \ldots, d_l\}$, where $m = (n-2)^{n-3}-1$ and $l = (n-3)(2^{n-3} - 1)$, and 
\be
\item $a_1 = {1 \choose n}{n-1 \choose n}(2,\ldots,n-2)$, $a_2 = {1 \choose n}{n-1 \choose n}(2,3)$, $a_3 = {1 \choose n}{n-1 \choose n}{n-2 \choose 2}$; 
\item For $1 \le i \le n-3$, $b_i = {1 \choose n}{n-1 \choose n}{i+1 \choose n-1}$; 
\item Each $c_i$ defines a distinct transformation in $\Uf^2_n$ other than $[n-1,n,\ldots,n,n]$; 
\item Each $d_i$ defines a distinct transformation in $\Uf^3_n$ other than $[j,n,\ldots,n,n]$ for all $j \in \{2,\ldots,n-2\}$. 
\ee
\end{proposition}

\begin{proof} 
Since $\Hbf(n) \subseteq \Wbf(n)$, we have $\langle \Hbf(n) \rangle \subseteq \Wbf(n)$. It remains to be shown that $\Wbf(n) \subseteq \langle \Hbf(n) \rangle$. Let $Q' = Q \setminus \{1,n-1,n\}$. 
\be 
\item First consider $\Uf^1_n$. By Theorem~\ref{thm:salomaa}, $a_1,a_2$ and $a_3$ together generate the semigroup $$\mathbf{Y}' = \{t \in \Uf^1_n \mid \txt{for all} i \in Q', it \in Q'\},$$ which is contained in $\Uf^1_n$. For any $t \in \Uf^1_n \setminus \mathbf{Y}'$, let $E_t = \{ i \in Q \mid it = n-1 \}$; then $E_t \neq \emptyset$. Suppose $E_t = \{i_1,\ldots,i_k\}$, where $1 \le k \le n-3$. Then there exists $t' \in \mathbf{Y}'$ such that, for all $i \not\in E_t$, $it' = it$. Let $s = b_{i_1-1} \cdots b_{i_k-1}$. Note that $E_ts = \{n-1\}$, and, for all $i \not\in E_t$, $i(t's) = (it')s = it$. So $t's = t$, and $\langle a_1,a_2,a_3,b_1,\ldots,b_{n-3} \rangle = \Uf^1_n$. 

\item Next, the transformations that are in $\Uf^2_n \cup \Uf^3_n$ but not in $\Hbf(n)$ are $t_i = [i,n,\ldots,n,n]$, where $2 \le i \le n-1$. Note that $d = {1 \choose 2}{n-1 \choose n}{Q' \choose n-1} \in \Hbf(n)$, and, for each $i \in \{2,\ldots,n-1\}$, $s_i = {1 \choose n}{n-1 \choose n}{2 \choose i} \in \Uf^1_n$. Then $t_i = ds_i \in \langle \Hbf(n) \rangle$, and $\Uf^2_n \cup \Uf^3_n \subseteq \langle \Hbf(n) \rangle$. 
\ee 

Therefore $\Wbf(n) = \langle \Hbf(n) \rangle$. \qed
\end{proof}

\begin{theorem}\label{thm:bfaut1} 
For $n \ge 5$, let $\cA'_n = (Q, \Sig, \delta, 1, F)$ be the DFA with alphabet $\Sig$ of size $(n-2)^{n-3} + (n-3)2^{n-3}+2$, where each letter defines a transformation as in Proposition~\ref{prop:Rgen}, and $F = \{n-1\}$. Then $L' = L(\cA'_n)$ has quotient complexity $\kappa(L') = n$, and syntactic complexity $\sigma(L') = \wbf(n)$. Moreover, $L'$ is bifix-free. 
\end{theorem}

\begin{proof} 
First, for all $i \in Q \setminus \{1\}$, there exists $a \in \Sig$ such that $t_a = [i, n, \ldots, n, n] \in \Hbf(n)$, and state $i$ is reachable by $a$. So all the states in $Q$ are reachable. Next, there exist $b, c \in \Sig$ such that $t_b = [n-1, n, \ldots, n, n] \in \Hbf(n)$ and $t_c = [n, 3, 4, \ldots, n, n] \in \Hbf(n)$. The initial state accepts $b$, while all other states reject it. For $2 \le i \le n-2$, state $i$ accepts $b^{n-i-1}$, while all other states reject it. Also, state $n-1$ is the only accepting state, and state $n$ is the empty state. Then all the states in $Q$ are distinct, and $\kappa(L') = n$. 

By Proposition~\ref{prop:Rgen}, the syntactic semigroup of $L'$ is $\Wbf(n)$; so $\sigma(L')=\wbf(n)$. By Proposition~\ref{prop:bf}, $L'$ is bifix-free. \qed
\end{proof}

\smallskip

\begin{conjecture}[Bifix-Free Regular Languages]
\label{con:bf}
If $L$ is a bifix-free regular language with $\kappa(L) = n \ge 6$, then $\sigma(L) \le \wbf(n)$ and this is a tight bound. 
\end{conjecture}

\smallskip

The conjecture holds for $n=6$ as we now show:
\begin{proof} 
When $n = 6$, $|\Bbf(6)| = 339$ and $|\Vbf(6)| = 213$. There are 126 transformations $\tau_1,\ldots,\tau_{126}$ in $\Bbf(6) \setminus \Vbf(6)$. For each $\tau_i$, we enumerated transformations in $\Wbf(6)$ using \emph{GAP} and found a unique $t_i \in \Vbf(6)$ such that $\langle t_i, \tau_i \rangle \not\subseteq \Bbf(6)$. Thus, for each $i$, at most one of $t_i$ and $\tau_i$ can appear in the syntactic semigroup $T_L$ of $L$. So we further lower the bound to $\sigma(L) \le 213$. This bound is reached by the DFA $\cA'_6$ in Theorem~\ref{thm:bfaut1}; so it is a tight upper bound for~$n~=~6$. \qed
\end{proof}

\section{Factor-Free Regular Languages}\label{sec:ff}

Let $L$ be a factor-free regular language with $\kappa(L) = n$. Since factor-free regular languages are also bifix-free, $L$ as a quotient is uniquely reachable, $\eps$ is the only accepting quotient of $L$, and $L$ also has the empty quotient. As in Section~\ref{sec:bf}, we assume that $Q$ is the set of states of quotient DFA of $L$, in which $1$ is the initial state, and states $n-1$ and $n$ correspond to the quotients $\eps$~and~$\emptyset$, respectively. Let 
$$\Bff(n) = \{t \in \Bbf(n) \mid \txt{for all} j \ge 1, 1t^j = n - 1 \Rightarrow it^j = n ~~\forall~i, 1 < i < n-1 \}.$$
We first have the following observation: 

\begin{proposition}\label{prop:ff}
If $L$ is a regular language with quotient complexity $n$ and syntactic semigroup $T_L$, then the following hold: 
\be 
\item If $L$ is factor-free, then $T_L$ is a subset of $\Bff(n)$. 
\item If $\eps$ is the only accepting quotient of $L$, and $T_L \subseteq \Bff(n)$, then $L$ is factor-free. 
\ee
\end{proposition}

\begin{proof}
1. Assume $L$ is factor-free. Then $L$ is bifix-free, and $T_L \subseteq \Bbf(n)$ by Proposition~\ref{prop:bf}. For any transformation $t_w \in T_L$ performed by some non-empty word $w$, if $1t^j_w = n - 1$ for some $j \ge 1$, then $w^j \in L$. If we also have $it^j_w \neq n$ for some $i \in Q \setminus \{1\}$, then $i \not\in \{n - 1, n\}$ as $(n-1)t = nt = n$ for all $t \in \Bff(n)$. Thus there exist non-empty words $u$ and $v$ such that state $i$ is reachable by $u$, and state $i(t^j_w)$ accepts $v$. So $uw^jv \in L$, which is a contradiction. Hence $T_L \subseteq \Bff(n)$. 

2. Since $\eps$ is the only accepting state and $\Bff(n) \subseteq \Bbf(n)$, $L$ is bifix-free by Proposition~\ref{prop:bf}. If $L$ is not factor-free, then there exist non-empty words $u, v$ and $w$ such that $w, uwv \in L$. Thus $1t_w = n - 1$, and $1t_{uwv} = 1(t_ut_wt_v) = n - 1$. Since $L$ is bifix-free, $1t_u \neq 1$ and $nt_v = n$; thus $(1t_u)t_w \neq n$, which contradicts the assumption that $t_w \in T_L \subseteq \Bff(n)$. Therefore $L$ is bifix-free. \qed
\end{proof}

The properties of suffix- and bifix-free regular languages still apply to factor-free regular languages. Moreover, we have

\begin{lemma}\label{lem:ffseq} 
For all $t \in \Bff(n)$ and $i \not\in \seq_t(1)$, if $n-1 \in \seq_t(1)$, then $n \in \seq_t(i)$. 
\end{lemma}

\begin{proof} 
Suppose $n-1 = 1t^k \in \seq_t(1)$ for some $k \ge 1$. If $n \not\in \seq_t(i)$, then for all $j \ge 1$, $it^j \neq n$. In particular, $it^k \neq n$, which contradicts the definition of $\Bff(n)$. Therefore $n \in \seq_t(i)$. \qed
\end{proof}

\begin{lemma}\label{lem:Fn} 
For $n \ge 3$, we have $|\Bff(n)| = N_n + O_n$, where 

\begin{eqnarray*} 
  O_n &=& 1 + \sum_{k=2}^{n-2} C^{n-3}_{k-1} (k-1)! \sum_{\begin{subarray}{c} 
      r_2 + \cdots + r_k + r \\ 
          = n-k-2
    \end{subarray}}
    C^{n-k-2}_{r_2,\ldots,r_k,r} S'_{r+1}(k) \prod_{j=2}^k S'_{r_j+1}(j-1), 
\end{eqnarray*}
and $N_n$ as given in Equation~(\ref{eq:H2}). 
\end{lemma}

\begin{proof} 
Let $t \in \Bff(n)$ be any transformation. Suppose $\seq_t(1) = 1,1t,\ldots,1t^k,n$, where $0 \le k \le n-2$. Then there are two cases: 

\be 
\item $n-1 \in \seq_t(1)$. Since $(n-1)t = n$, we have $n-1 = 1t^k$, and $k \ge 1$. If $k = 1$, then $1t = n-1$, and $it = n$ for all $i \neq 1$; such a $t$ is unique. Consider $k \ge 2$. There are $C^{n-2}_{k-1} (k-1)!$ different $\seq_t(1)$. For $2 \le j \le k$, suppose there are $r_j + 1$ nodes in tree $T_t(j)$; then there are $S'_{r_j+1}(j-1)$ such trees. Let $E$ be the set of elements $x$ that are not in any tree $T_t(j)$ nor in $\seq_t(1)$, and let $r = |E| = (n-k-2) - (r_2 + \cdots + r_k)$. By Lemma~\ref{lem:ffseq}, $n \in \seq_t(x)$ for all $x \in E$. Then the union of paths $\tpath_t(x)$ for all $x \in E$ form a labeled tree rooted at $n$ with height at most $k$, and there are $S'_{r+1}(k)$ such trees. Thus the number of transformations in this case is $O_n$. 

\item $n-1 \not\in \seq_t(1)$. Now, for all $j \ge 1$, $1t^j \neq n-1$. Then $t \in \Bbf(n)$. As in the proof of Lemma~\ref{lem:Hn}, the number of transformations in this case is $N_n$. 

\ee

Altogether we have the desired formula. \qed
\end{proof}

Let $\bff(n) = |\Bff(n)|$. From Proposition~\ref{prop:ff} and Lemma~\ref{lem:Fn} we have 
\begin{proposition}\label{prop:Fncard} 
For $n \ge 3$, if $L$ is a factor-free regular language with quotient complexity $n$, then its syntactic complexity $\sigma(L)$ satisfies that $\sigma(L) \le \bff(n)$, where $\bff(n)$ is the cardinality of $\Bff(n)$ as in Lemma~\ref{lem:Fn}.
\end{proposition} 

The tight upper bound on the syntactic complexity of factor-free regular languages is reached by the largest semigroup contained in $\Bff(n)$. When $2 \le n \le 4$, $\Bff(n)$ is a semigroup. The languages $L_2 = \eps$, $L_3 = a$ over alphabet $\{a,b\}$, and $L_4 = ab^*a$ have syntactic complexities $1 = \bff(2)$, $2 = \bff(3)$, and $6 = \bff(4)$, respectively. So $\bff(n)$ is a tight upper bound for $n \in \{2,3,4\}$. However, the set $\Bff(n)$ is not a semigroup for $n \ge 5$, because $s_1 = [2,3,\ldots,n-1,n,n], s_2 = {n-1 \choose n}{2 \choose n-1}{1 \choose n} = [n,n-1,3,\ldots,n-2,n,n] \in \Bff(n)$ but $s_1s_2 = [n-1,3,\ldots,n-2,n,n,n] \not\in \Bff(n)$.

Next, we find a large semigroup that can be the syntactic semigroup of a factor-free regular language. 

Let $t_0 = {Q \setminus \{1\} \choose n}{1 \choose n-1} = [n-1, n, \ldots, n]$, and let $\Wff(n) = \Uf^1_n \cup \{t_0\} \cup \Uf^3_n$. When $2 \le n \le 4$, we have $\Wff(n) = \Bff(n)$. So we are interested in larger values of $n$ in the rest of this section. 

\begin{proposition}\label{prop:Uncard}
For $n \ge 5$, $\Wff(n)$ is a semigroup contained in $\Bff(n)$ with cardinality 
$$\wff(n) = |\Wff(n)| = (n-1)^{n-3} + (n-3)2^{n-3} + 1.$$
\end{proposition}

\begin{proof}

As we have shown in the proof of Proposition~\ref{prop:Rcard}, $\Uf^1_n$ is a semigroup. For any $t \in \Uf^1_n \cup \{t_0\}$, since $t_0 \in \Uf^2_n$, we have $tt_0, t_0t \in \Uf^1_n$; so $\Uf^1_n \cup \{t_0\}$ is also a semigroup. We also know that, for any $t_3 \in \Uf^3_n$ and $t' \in \Wff(n)$, since $\Wff(n) \subseteq \Wbf(n)$, $i(t_3t') = n$ for all $i \neq 1$; so $t_3t' \in \Wff(n)$. If $t' \in \Uf^1_n \cup \{t_0\}$, then $1t't_3 = n$ and $t't_3 \in \Uf^1_n$; otherwise, $t' \in \Uf^3_n$, and $t't_3 = t_2$ or ${Q \choose n} \in \Uf^1_n$. Hence $\Wff(n)$ is a semigroup. 

For any $t \in \Uf^1_n$, since $1t = n$, we have $t \in \Bff(n)$. For any $t \in \Uf^3_n$, $1t \neq n-1$, and $it^2 = n$ for all $i \in \{2,\ldots,n\}$; then $t \in \Bff(n)$ as well. Clearly $t_0 \in \Bff(n)$. Hence $\Wff(n)$ is contained in $\Bff(n)$. 

We know that $|\Uf^1_n| = (n-1)^{n-3}$ and $|\Uf^3_n| = (n-3)2^{n-3}$. Therefore $|\Wff(n)| = (n-1)^{n-3} + (n-3)2^{n-3} + 1$. \qed
\end{proof}

We now describe a generating set of $\Wff(n)$. 
\begin{proposition}\label{prop:Ugen} 
For $n \ge 5$, the semigroup $\Wff(n)$ is generated by $\Hff(n) = \{a_1, a_2, a_3, b_1, \ldots, b_{n-3}, c_1, \ldots, c_m\}$, where $m = (n-3)(2^{n-3}-1)$, and 
\be
\item $a_1 = {1 \choose n}{n-1 \choose n}(2,\ldots,n-2)$, $a_2 = {1 \choose n}{n-1 \choose n}(2,3)$, $a_3 = {1 \choose n}{n-1 \choose n}{n-2 \choose 2}$; 
\item For $1 \le i \le n-3$, $b_i = {1 \choose n}{n-1 \choose n}{i+1 \choose n-1}$; 
\item Each $c_i$ defines a distinct transformation in $\Uf^3_n$ other than $[j,n,\ldots,n,n]$ for all $j \in \{2,\ldots,n-2\}$.
\ee
\end{proposition}

\begin{proof} 
We know from the proof of Proposition~\ref{prop:Rgen} that $\Uf^1_n$ is generated by $\{a_1,a_2,a_3,b_1,\ldots,b_{n-3}\}$. Also, the transformations that are in $\{t_0\} \cup \Uf^3_n$ but not in $\Hff(n)$ are $t_j = [j,n,\ldots,n,n]$, where $j \in \{2,\ldots,n-1\}$. Each $t_j$ is a composition of $d = {1 \choose 2}{n-1 \choose n}{Q' \choose n-1} \in \Hbf(n)$ and $s_j = {1 \choose n}{n-1 \choose n}{2 \choose j} \in \Uf^1_n$. Therefore $\langle \Hff(n) \rangle = \Wff(n)$. \qed
\end{proof}

\begin{theorem}\label{thm:ffaut} 
For $n \ge 5$, let $\cA_n = (Q, \Sig, \delta, 1, F)$ be the DFA with alphabet $\Sig = \{a_1, a_2, a_3, b_1, \ldots, b_{n-3}, c_1, \ldots, c_m\}$ of size $(n-3)2^{n-3}+3$, where each letter defines a transformation as in Proposition~\ref{prop:Ugen}, and $F = \{n-1\}$. Then $L = L(\cA_n)$ has quotient complexity $\kappa(L) = n$, and syntactic complexity $\sigma(L)~=~\wff(n)$. Moreover, $L$ is factor-free. 
\end{theorem}

\begin{proof} 
Since $\Hff(n) \subseteq \Hbf(n)$, the DFA $\cA_n$ can be obtained from the DFA $\cA'_n$ of Theorem~\ref{thm:bfaut1} by restricting the alphabet. The words used to show that all the states of $\cA'$ are reachable and distinct still exist in $\cA_n$. Then we have $\kappa(L) = n$. By Proposition~\ref{prop:Ugen}, the syntactic semigroup of $L$ is $\Wff(n)$; so $\sigma(L) = \wff(n)$. By Proposition~\ref{prop:ff}, $L$ is factor-free. \qed
\end{proof}

\begin{conjecture}[Factor-Free Regular Languages]\label{con:ff}
If $L$ is a factor-free regular language with $\kappa(L) = n$, where $n \ge 5$, then $\sigma(L) \le \wff(n)$ and this is a tight upper~bound. 
\end{conjecture}

We prove the conjecture for $n = 5$ and $6$. 

\begin{proof} 
For $n = 5$, $|\Bff(5)| = 31$, and $|\Wff(5)| = 25$. There are 6 transformations $\tau_1, \ldots, \tau_{6}$ in $\Bff(5) \setminus \Wff(5)$. For each $\tau_i$, $1 \le i \le 6$, we found a unique $t_i \in \Wff(5)$ such that $\langle t_i, \tau_i \rangle \not\subseteq \Bff(5)$: 
$$\begin{array}{ll}
  \tau_1 = [ 2, 3, 4, 5, 5 ], \quad & \quad t_1 = [ 5, 2, 2, 5, 5 ], \\ 
  \tau_2 = [ 2, 3, 5, 5, 5 ], \quad & \quad t_2 = [ 5, 4, 2, 5, 5 ], \\ 
  \tau_3 = [ 2, 5, 3, 5, 5 ], \quad & \quad t_3 = [ 5, 3, 3, 5, 5 ], \\ 
  \tau_4 = [ 3, 2, 5, 5, 5 ], \quad & \quad t_4 = [ 5, 2, 4, 5, 5 ], \\ 
  \tau_5 = [ 3, 4, 2, 5, 5 ], \quad & \quad t_5 = [ 5, 3, 2, 5, 5 ], \\ 
  \tau_6 = [ 3, 5, 2, 5, 5 ], \quad & \quad t_6 = [ 5, 3, 4, 5, 5 ]. 
\end{array}$$ 

For each $1 \le i \le 6$, at most one of $t_i$ and $\tau_i$ can appear in the syntactic semigroup $T_L$ of a factor-free regular language $L$. Then $\sigma(L) = |T_L| \le 25$. By Theorem~\ref{thm:ffaut}, this upper bound is tight for $n = 5$. 

For $n = 6$, $|\Bff(6)| = 246$, and $|\Wff(6)| = 150$. There are 96 transformations $\tau_1, \ldots, \tau_{96}$ in $\Bff(6) \setminus \Wff(6)$. For each $\tau_i$, $1 \le i \le 72$, we enumerated the transformations in $\Wff(6)$ using \emph{GAP} and found a unique $t_i \in \Wff(6)$ such that $\langle t_i, \tau_i \rangle \not\subseteq \Bff(6)$. Thus $150$ is a tight upper bound for $n = 6$. \qed
\end{proof}

\section{Quotient Complexity of the Reversal of Free Languages}\label{sec:rev}

It has been shown in~\cite{BrYe11} that for certain regular languages with maximal syntactic complexity, the reverse languages have maximal quotient complexity. This is also true for some free languages, as we now show. 

In this section we consider \emph{non-deterministic finite automata} (NFA). A NFA $\cN$ is a quintuple $\cN = (Q, \Sig, \delta, I, F)$, where $Q$, $\Sig$, and $F$ are as in a DFA, $\delta : Q \times \Sig \to 2^Q$ is the non-deterministic transition function, and $I$ is the set of initial states. For any word $w \in \Sig^*$, the \emph{reverse} of $w$ is defined inductively as follows: $w^R = \eps$ if $w = \eps$, and $w^R = u^Ra$ if $w = au$ for some $a \in \Sig$ and $u \in \Sig^*$. The \emph{reverse} of any language $L$ is the language $L^R = \{w^R \mid w \in L\}$. For any finite automaton (DFA or NFA) $\cM$, we denote using $\cM^R$ the automaton obtained by reversing $\cM$ and exchanging the roles of initial states and accepting states, and $\cM^D$, the DFA obtained by applying the subset construction to $\cM$. Then $L(\cM^R) = (L(\cM))^R$, and $L(\cM^D) = L(\cM)$. To simplify our proofs, we use an observation from~\cite{Brz62} that, for any NFA $\cN$ whose states are all reachable, if the automaton $\cN^R$ is deterministic, then the DFA $\cN^D$ is minimal. 

\begin{theorem}\label{thm:pfrev}
The reverse of the prefix-free regular language accepted by 
the DFA $\cA_n$ of Theorem~\ref{thm:prefix-free} restricted to $\{a,c,d_{n-2}\}$
has $2^{n-2}+1$ quotients, which is the maximum possible for a prefix-free regular language. 
\end{theorem}

\begin{proof}\label{proof:pfrev}
Let $\cB_n$ be the DFA $\cA_n$ restricted to $\{a,c,d_{n-2}\}$. Since $L(\cA_n)$ is prefix-free, so is $L_n = L(\cB_n)$. We show that  $\kappa(L_n^R) = 2^{n-2}+1$. 

Let $\cN_n$ be the NFA obtained by removing unreachable states from the NFA $\cA_n^R$. (See Fig.~\ref{fig:pfrev} for $\cN_6$.)  We first prove that the following $2^{n-2}+1$ sets of states of $\cN_n$ are reachable:
$\{\{n-1\}\} \cup \{S \mid S \subseteq \{1,\ldots,n-2\}~\}.$

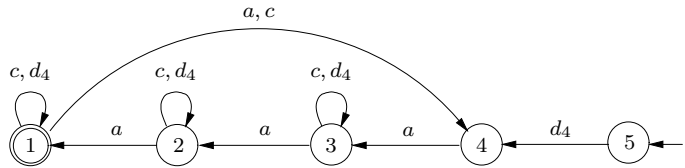
\begin{figure}[hbt]
\begin{center}
\setlength{\unitlength}{0.00056868in}
\begingroup\makeatletter\ifx\SetFigFont\undefined%
\gdef\SetFigFont#1#2#3#4#5{%
  \reset@font\fontsize{#1}{#2pt}%
  \fontfamily{#3}\fontseries{#4}\fontshape{#5}%
  \selectfont}%
\fi\endgroup%
{\renewcommand{\dashlinestretch}{30}
\begin{picture}(6256,1578)(0,-10)
\put(5758,161){\makebox(0,0)[b]{\smash{{\SetFigFont{8}{9.6}{\familydefault}{\mddefault}{\updefault}5}}}}
\put(199,131){\makebox(0,0)[b]{\smash{{\SetFigFont{8}{9.6}{\familydefault}{\mddefault}{\updefault}1}}}}
\put(1558,199){\ellipse{382}{382}}
\put(1572,131){\makebox(0,0)[b]{\smash{{\SetFigFont{8}{9.6}{\familydefault}{\mddefault}{\updefault}2}}}}
\put(2988,204){\ellipse{382}{382}}
\put(2989,131){\makebox(0,0)[b]{\smash{{\SetFigFont{8}{9.6}{\familydefault}{\mddefault}{\updefault}3}}}}
\put(4383,204){\ellipse{382}{382}}
\put(4384,131){\makebox(0,0)[b]{\smash{{\SetFigFont{8}{9.6}{\familydefault}{\mddefault}{\updefault}4}}}}
\put(2314.000,-1170.286){\arc{4901.878}{3.8024}{5.6224}}
\blacken\path(4149.913,408.039)(4249.000,334.000)(4196.413,445.958)(4149.913,408.039)
\put(201,194){\ellipse{324}{324}}
\put(5743,222){\ellipse{382}{382}}
\blacken\path(3311.000,229.000)(3191.000,199.000)(3311.000,169.000)(3311.000,229.000)
\path(3191,199)(4172,199)
\blacken\path(1894.000,229.000)(1774.000,199.000)(1894.000,169.000)(1894.000,229.000)
\path(1774,199)(2781,199)
\blacken\path(506.000,229.000)(386.000,199.000)(506.000,169.000)(506.000,229.000)
\path(386,199)(1367,199)
\blacken\path(4698.000,237.000)(4578.000,207.000)(4698.000,177.000)(4698.000,237.000)
\path(4578,207)(5559,207)
\path(6244,207)(5943,207)
\blacken\path(6063.000,237.000)(5943.000,207.000)(6063.000,177.000)(6063.000,237.000)
\path(109,379)(108,382)(105,388)
	(101,398)(95,412)(88,429)
	(81,448)(74,468)(68,489)
	(63,511)(59,534)(58,557)
	(59,581)(64,604)(71,622)
	(79,638)(87,650)(94,659)
	(101,665)(106,670)(112,673)
	(117,675)(122,677)(128,679)
	(135,681)(143,684)(154,687)
	(167,690)(182,693)(199,694)
	(216,693)(231,690)(244,687)
	(255,684)(263,681)(270,679)
	(276,677)(282,675)(286,673)
	(292,670)(297,665)(304,659)
	(311,650)(319,638)(327,622)
	(334,604)(339,581)(340,557)
	(339,534)(335,511)(330,489)
	(324,468)(317,448)(310,429)
	(303,412)(289,379)
\blacken\path(308.249,501.186)(289.000,379.000)(363.483,477.753)(308.249,501.186)
\path(2899,379)(2898,382)(2895,388)
	(2891,398)(2885,412)(2878,429)
	(2871,448)(2864,468)(2858,489)
	(2853,511)(2849,534)(2848,557)
	(2849,581)(2854,604)(2861,622)
	(2869,638)(2877,650)(2884,659)
	(2891,665)(2896,670)(2902,673)
	(2907,675)(2912,677)(2918,679)
	(2925,681)(2933,684)(2944,687)
	(2957,690)(2972,693)(2989,694)
	(3006,693)(3021,690)(3034,687)
	(3045,684)(3053,681)(3060,679)
	(3066,677)(3072,675)(3076,673)
	(3082,670)(3087,665)(3094,659)
	(3101,650)(3109,638)(3117,622)
	(3124,604)(3129,581)(3130,557)
	(3129,534)(3125,511)(3120,489)
	(3114,468)(3107,448)(3100,429)
	(3093,412)(3079,379)
\blacken\path(3098.249,501.186)(3079.000,379.000)(3153.483,477.753)(3098.249,501.186)
\path(1459,379)(1458,382)(1455,388)
	(1451,398)(1445,412)(1438,429)
	(1431,448)(1424,468)(1418,489)
	(1413,511)(1409,534)(1408,557)
	(1409,581)(1414,604)(1421,622)
	(1429,638)(1437,650)(1444,659)
	(1451,665)(1456,670)(1462,673)
	(1467,675)(1472,677)(1478,679)
	(1485,681)(1493,684)(1504,687)
	(1517,690)(1532,693)(1549,694)
	(1566,693)(1581,690)(1594,687)
	(1605,684)(1613,681)(1620,679)
	(1626,677)(1632,675)(1636,673)
	(1642,670)(1647,665)(1654,659)
	(1661,650)(1669,638)(1677,622)
	(1684,604)(1689,581)(1690,557)
	(1689,534)(1685,511)(1680,489)
	(1674,468)(1667,448)(1660,429)
	(1653,412)(1639,379)
\blacken\path(1658.249,501.186)(1639.000,379.000)(1713.483,477.753)(1658.249,501.186)
\put(1549,829){\makebox(0,0)[b]{\smash{{\SetFigFont{8}{9.6}{\familydefault}{\mddefault}{\updefault}$c, d_4$}}}}
\put(2989,829){\makebox(0,0)[b]{\smash{{\SetFigFont{8}{9.6}{\familydefault}{\mddefault}{\updefault}$c, d_4$}}}}
\put(199,829){\makebox(0,0)[b]{\smash{{\SetFigFont{8}{9.6}{\familydefault}{\mddefault}{\updefault}$c,d_4$}}}}
\put(2314,1392){\makebox(0,0)[b]{\smash{{\SetFigFont{8}{9.6}{\familydefault}{\mddefault}{\updefault}$a,c$}}}}
\put(3702,281){\makebox(0,0)[b]{\smash{{\SetFigFont{8}{9.6}{\familydefault}{\mddefault}{\updefault}$a$}}}}
\put(2374,296){\makebox(0,0)[b]{\smash{{\SetFigFont{8}{9.6}{\familydefault}{\mddefault}{\updefault}$a$}}}}
\put(1001,289){\makebox(0,0)[b]{\smash{{\SetFigFont{8}{9.6}{\familydefault}{\mddefault}{\updefault}$a$}}}}
\put(5119,295){\makebox(0,0)[b]{\smash{{\SetFigFont{8}{9.6}{\familydefault}{\mddefault}{\updefault}$d_4$}}}}
\put(199,199){\ellipse{382}{382}}
\end{picture}
}
\end{center}
\caption{NFA $\cN_6$ of $L_6^R$ with quotient complexity $\kappa(L_6^R) = 17$; empty state omitted.}
\label{fig:pfrev}
\end{figure}

The singleton set $\{n-1\}$ of initial states  of $\cN_n$ is reached by $\eps$. From $\{n-1\}$ we reach the empty set by $a$. 
The set $\{n-2\}$ is reached by $d_{n-2}$ from $\{n-1\}$, and from here, $\{1\}$ is reached by $a^{n-3}$. From any set $\{1,2,\ldots,i\}$, where $1\le i<n-2$, we reach 
$\{1,2,\ldots,i, i+1\}$ by $ca^{n-3}$. Thus we reach $\{1,2,\ldots,n-2\}$ from $\{1\}$ by
$(ca^{n-3})^{n-3}$.
Now assume that any set $S$ of cardinality $l\le n-2$ can be reached; then we can get a set of cardinality $l-1$ by deleting an element~$j$ from $S$ by applying 
$a^jd_{n-2}a^{n-2-j}$. Hence all the subsets of  $\{1,2,\ldots,n-2\}$ can be reached.

The automaton $\cN_n^R$ is a subset of $\cA_n$, and it is deterministic. Then $\cN_n^D$ is minimal. Hence $\kappa(L_n^R) = 2^{n-2}+1$, which is the maximal quotient complexity of reversal of prefix-free languages as shown in~\cite{HSW09}. \qed
\end{proof}

\medskip
It is interesting that, for suffix-, bifix-, and factor-free regular languages, although we don't have tight upper bounds on their syntactic complexities, some languages in these classes with large syntactic complexities have their reverse languages reaching the upper bounds on the quotient complexities for the reversal operation. 

\begin{theorem}\label{thm:sfrev}
The reverse of the suffix-free regular language accepted by the DFA $\cA_n'$ of Theorem~\ref{thm:wsfaut} restricted to $\{a_1,a_2,a_3,c\}$ has $2^{n-2}+1$ quotients, which is the maximum possible for a suffix-free regular language. 
\end{theorem}

\begin{proof}
Let $\cC_n$ be the DFA $\cA_n'$ restricted to the alphabet $\{a_1,a_2,a_3,c\}$. Since $L(\cA_n')$ is suffix-free, so is $L_n' = L(\cC_n)$. Let $\cN_n'$ be the NFA obtained from $\cC_n^R$ by removing unreachable states. Figure~\ref{fig:sfrev} shows the NFA $\cN_6'$. 

\begin{figure}[hbt]
\begin{center}
\setlength{\unitlength}{0.00052493in}
\begingroup\makeatletter\ifx\SetFigFont\undefined%
\gdef\SetFigFont#1#2#3#4#5{%
  \reset@font\fontsize{#1}{#2pt}%
  \fontfamily{#3}\fontseries{#4}\fontshape{#5}%
  \selectfont}%
\fi\endgroup%
{\renewcommand{\dashlinestretch}{30}
\begin{picture}(5986,2010)(0,-10)
\put(2318.000,1182.000){\arc{1530.000}{1.0808}{2.0608}}
\blacken\path(2579.885,431.678)(2678.000,507.000)(2555.982,486.711)(2579.885,431.678)
\put(1643.000,968.250){\arc{337.500}{2.4981}{6.9267}}
\blacken\path(1781.575,990.641)(1778.000,867.000)(1839.339,974.413)(1781.575,990.641)
\put(2993.000,968.250){\arc{337.500}{2.4981}{6.9267}}
\blacken\path(3131.575,990.641)(3128.000,867.000)(3189.339,974.413)(3131.575,990.641)
\put(4343.000,968.250){\arc{337.500}{2.4981}{6.9267}}
\blacken\path(4481.575,990.641)(4478.000,867.000)(4539.339,974.413)(4481.575,990.641)
\put(5693.000,968.250){\arc{337.500}{2.4981}{6.9267}}
\blacken\path(5831.575,990.641)(5828.000,867.000)(5889.339,974.413)(5831.575,990.641)
\put(3684.189,-495.622){\arc{4435.362}{3.7528}{5.6470}}
\blacken\path(5370.521,898.143)(5468.000,822.000)(5417.821,935.058)(5370.521,898.143)
\put(2993,597){\ellipse{570}{570}}
\put(293,597){\ellipse{570}{570}}
\put(1643,597){\ellipse{570}{570}}
\put(4343,597){\ellipse{570}{570}}
\put(5693,597){\ellipse{570}{570}}
\put(293,597){\ellipse{450}{450}}
\path(1373,597)(608,597)
\blacken\path(728.000,627.000)(608.000,597.000)(728.000,567.000)(728.000,627.000)
\path(2723,597)(1958,597)
\blacken\path(2078.000,627.000)(1958.000,597.000)(2078.000,567.000)(2078.000,627.000)
\path(1643,12)(1643,282)
\blacken\path(1673.000,162.000)(1643.000,282.000)(1613.000,162.000)(1673.000,162.000)
\path(4073,597)(3308,597)
\blacken\path(3428.000,627.000)(3308.000,597.000)(3428.000,567.000)(3428.000,627.000)
\path(5423,597)(4658,597)
\blacken\path(4778.000,627.000)(4658.000,597.000)(4778.000,567.000)(4778.000,627.000)
\put(293,539){\makebox(0,0)[b]{\smash{{\SetFigFont{7}{8.4}{\familydefault}{\mddefault}{\updefault}1}}}}
\put(1643,539){\makebox(0,0)[b]{\smash{{\SetFigFont{7}{8.4}{\familydefault}{\mddefault}{\updefault}2}}}}
\put(2993,539){\makebox(0,0)[b]{\smash{{\SetFigFont{7}{8.4}{\familydefault}{\mddefault}{\updefault}3}}}}
\put(4343,539){\makebox(0,0)[b]{\smash{{\SetFigFont{7}{8.4}{\familydefault}{\mddefault}{\updefault}4}}}}
\put(5693,539){\makebox(0,0)[b]{\smash{{\SetFigFont{7}{8.4}{\familydefault}{\mddefault}{\updefault}5}}}}
\put(1643,1272){\makebox(0,0)[b]{\smash{{\SetFigFont{9}{10.8}{\familydefault}{\mddefault}{\updefault}$a_3$}}}}
\put(2318,192){\makebox(0,0)[b]{\smash{{\SetFigFont{9}{10.8}{\familydefault}{\mddefault}{\updefault}$a_2$}}}}
\put(2993,1272){\makebox(0,0)[b]{\smash{{\SetFigFont{9}{10.8}{\familydefault}{\mddefault}{\updefault}$a_3$}}}}
\put(4343,1272){\makebox(0,0)[b]{\smash{{\SetFigFont{9}{10.8}{\familydefault}{\mddefault}{\updefault}$a_2,a_3$}}}}
\put(5693,1272){\makebox(0,0)[b]{\smash{{\SetFigFont{9}{10.8}{\familydefault}{\mddefault}{\updefault}$a_2$}}}}
\put(3668,665){\makebox(0,0)[b]{\smash{{\SetFigFont{9}{10.8}{\familydefault}{\mddefault}{\updefault}$a_1$}}}}
\put(5018,665){\makebox(0,0)[b]{\smash{{\SetFigFont{9}{10.8}{\familydefault}{\mddefault}{\updefault}$a_1$}}}}
\put(3668,1812){\makebox(0,0)[b]{\smash{{\SetFigFont{9}{10.8}{\familydefault}{\mddefault}{\updefault}$a_1,a_3$}}}}
\put(2318,665){\makebox(0,0)[b]{\smash{{\SetFigFont{9}{10.8}{\familydefault}{\mddefault}{\updefault}$a_1,a_2$}}}}
\put(968,665){\makebox(0,0)[b]{\smash{{\SetFigFont{9}{10.8}{\familydefault}{\mddefault}{\updefault}$c$}}}}
\end{picture}
}
\end{center}
\caption{NFA $\cN_6'$ of $L_6'^R$ with quotient complexity $\kappa(L_6'^R) = 17$; empty state omitted.}
\label{fig:sfrev}
\end{figure}
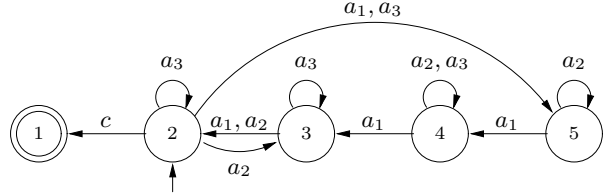

Apply the subset construction to $\cN_n'$, we get a DFA $\cN_n'^D$. Its initial state is a singleton set $\{2\}$. From the initial state, we can reach state $\{2,3,\ldots,i\}$ by $(a_3a_1^{n-3})^{i-2}$, where $3 \le i \le n-1$. Then the state $\{2,3,\ldots,n-1\}$ is reached from $\{2\}$ by $(a_3a_1^{n-3})^{n-3}$. Assume that any set $S$ of cardinality $l$ can be reached, where $2 \le l \le n-2$. If $j \in S$, then we can reach $S' = S \setminus \{j\}$ from $S$ by $a_1^{j-1}a_3a_1^{n-j-1}$. So all the nonempty subsets of $\{2,3,\ldots,n-1\}$ can be reached. We can also reach the singleton set $\{1\}$ from $\{2\}$ by $c$, and, from there, the empty state by $c$ again. Hence $\cN_n'^D$ has $2^{n-2}+1$ reachable states.

Since the automaton $\cN_n'^R$, the reverse of $\cN_n'$, is a subset of $\cC_n$, it is deterministic; hence $\cN_n'^D$ is minimal. Then the quotient complexity of $L_n'^R$ is $2^{n-2}+1$, which meets the upper bound for reversal of suffix-free regular languages~\cite{HS09}. \qed
\end{proof}

\begin{theorem}\label{thm:ffrev}
The reverse of the factor-free regular language accepted by the DFA $\cA_n$ of Theorem~\ref{thm:ffaut} restricted to the alphabet $\{a_1,a_2,a_3,c\}$, where $c = [2,n-1,n,\ldots,n,n] \in \Hff(n)$, has $2^{n-3}+2$ quotients, which is the maximum possible for a bifix- or factor-free regular language. 
\end{theorem}

\begin{proof}
Let $\cD_n$ be the DFA $\cA_n$ restricted to the alphabet $\{a_1,a_2,a_3,c\}$; then $L_n'' = L(\cD_n)$ is factor-free. Let $\cN_n''$ be the NFA obtained from $\cD_n^R$ by removing unreachable states. An example of $\cN_n''$ is shown in Figure~\ref{fig:ffrev}. 

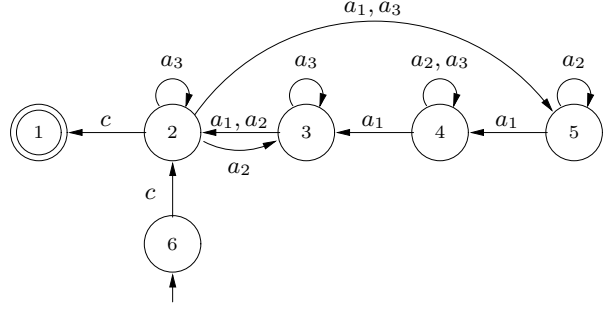
\begin{figure}[hbt]
\begin{center}
\setlength{\unitlength}{0.00052493in}
\begingroup\makeatletter\ifx\SetFigFont\undefined%
\gdef\SetFigFont#1#2#3#4#5{%
  \reset@font\fontsize{#1}{#2pt}%
  \fontfamily{#3}\fontseries{#4}\fontshape{#5}%
  \selectfont}%
\fi\endgroup%
{\renewcommand{\dashlinestretch}{30}
\begin{picture}(5986,3135)(0,-10)
\put(293,1722){\ellipse{570}{570}}
\put(293,1722){\ellipse{450}{450}}
\put(1643,597){\ellipse{570}{570}}
\put(2318.000,2307.000){\arc{1530.000}{1.0808}{2.0608}}
\blacken\path(2579.885,1556.678)(2678.000,1632.000)(2555.982,1611.711)(2579.885,1556.678)
\put(1643.000,2093.250){\arc{337.500}{2.4981}{6.9267}}
\blacken\path(1781.575,2115.641)(1778.000,1992.000)(1839.339,2099.413)(1781.575,2115.641)
\put(2993.000,2093.250){\arc{337.500}{2.4981}{6.9267}}
\blacken\path(3131.575,2115.641)(3128.000,1992.000)(3189.339,2099.413)(3131.575,2115.641)
\put(4343.000,2093.250){\arc{337.500}{2.4981}{6.9267}}
\blacken\path(4481.575,2115.641)(4478.000,1992.000)(4539.339,2099.413)(4481.575,2115.641)
\put(5693.000,2093.250){\arc{337.500}{2.4981}{6.9267}}
\blacken\path(5831.575,2115.641)(5828.000,1992.000)(5889.339,2099.413)(5831.575,2115.641)
\put(3684.189,629.378){\arc{4435.362}{3.7528}{5.6470}}
\blacken\path(5370.521,2023.143)(5468.000,1947.000)(5417.821,2060.058)(5370.521,2023.143)
\put(2993,1722){\ellipse{570}{570}}
\put(4343,1722){\ellipse{570}{570}}
\put(5693,1722){\ellipse{570}{570}}
\put(1643,1722){\ellipse{570}{570}}
\path(1373,1722)(608,1722)
\blacken\path(728.000,1752.000)(608.000,1722.000)(728.000,1692.000)(728.000,1752.000)
\path(2723,1722)(1958,1722)
\blacken\path(2078.000,1752.000)(1958.000,1722.000)(2078.000,1692.000)(2078.000,1752.000)
\path(4073,1722)(3308,1722)
\blacken\path(3428.000,1752.000)(3308.000,1722.000)(3428.000,1692.000)(3428.000,1752.000)
\path(5423,1722)(4658,1722)
\blacken\path(4778.000,1752.000)(4658.000,1722.000)(4778.000,1692.000)(4778.000,1752.000)
\path(1643,12)(1643,282)
\blacken\path(1673.000,162.000)(1643.000,282.000)(1613.000,162.000)(1673.000,162.000)
\path(1643,867)(1643,1407)
\blacken\path(1673.000,1287.000)(1643.000,1407.000)(1613.000,1287.000)(1673.000,1287.000)
\put(293,1664){\makebox(0,0)[b]{\smash{{\SetFigFont{7}{8.4}{\familydefault}{\mddefault}{\updefault}1}}}}
\put(1643,539){\makebox(0,0)[b]{\smash{{\SetFigFont{7}{8.4}{\familydefault}{\mddefault}{\updefault}6}}}}
\put(1643,1664){\makebox(0,0)[b]{\smash{{\SetFigFont{7}{8.4}{\familydefault}{\mddefault}{\updefault}2}}}}
\put(2993,1664){\makebox(0,0)[b]{\smash{{\SetFigFont{7}{8.4}{\familydefault}{\mddefault}{\updefault}3}}}}
\put(4343,1664){\makebox(0,0)[b]{\smash{{\SetFigFont{7}{8.4}{\familydefault}{\mddefault}{\updefault}4}}}}
\put(5693,1664){\makebox(0,0)[b]{\smash{{\SetFigFont{7}{8.4}{\familydefault}{\mddefault}{\updefault}5}}}}
\put(1643,2397){\makebox(0,0)[b]{\smash{{\SetFigFont{9}{10.8}{\familydefault}{\mddefault}{\updefault}$a_3$}}}}
\put(2318,1317){\makebox(0,0)[b]{\smash{{\SetFigFont{9}{10.8}{\familydefault}{\mddefault}{\updefault}$a_2$}}}}
\put(2993,2397){\makebox(0,0)[b]{\smash{{\SetFigFont{9}{10.8}{\familydefault}{\mddefault}{\updefault}$a_3$}}}}
\put(4343,2397){\makebox(0,0)[b]{\smash{{\SetFigFont{9}{10.8}{\familydefault}{\mddefault}{\updefault}$a_2,a_3$}}}}
\put(5693,2397){\makebox(0,0)[b]{\smash{{\SetFigFont{9}{10.8}{\familydefault}{\mddefault}{\updefault}$a_2$}}}}
\put(3668,1790){\makebox(0,0)[b]{\smash{{\SetFigFont{9}{10.8}{\familydefault}{\mddefault}{\updefault}$a_1$}}}}
\put(5018,1790){\makebox(0,0)[b]{\smash{{\SetFigFont{9}{10.8}{\familydefault}{\mddefault}{\updefault}$a_1$}}}}
\put(3668,2937){\makebox(0,0)[b]{\smash{{\SetFigFont{9}{10.8}{\familydefault}{\mddefault}{\updefault}$a_1,a_3$}}}}
\put(2318,1790){\makebox(0,0)[b]{\smash{{\SetFigFont{9}{10.8}{\familydefault}{\mddefault}{\updefault}$a_1,a_2$}}}}
\put(968,1790){\makebox(0,0)[b]{\smash{{\SetFigFont{9}{10.8}{\familydefault}{\mddefault}{\updefault}$c$}}}}
\put(1418,1047){\makebox(0,0)[b]{\smash{{\SetFigFont{9}{10.8}{\familydefault}{\mddefault}{\updefault}$c$}}}}
\end{picture}
}
\end{center}
\caption{NFA $\cN_7''$ of $L_7''^R$ with quotient complexity $\kappa(L_7''^R) = 18$; empty state omitted.}
\label{fig:ffrev}
\end{figure}

Note that $\cN_n''$ can be obtained from the NFA $\cN_{n-1}'$ in Theorem~\ref{thm:sfrev} by adding a new state $n-1$, which is the only initial state in $\cN_n''$, and the transition from $\{n-1\}$ to $\{2\}$ under input $c$. We know that all non-empty subsets of $\{2,3,\ldots,n-2\}$ are reachable from $\{2\}$. The accepting state $\{1\}$ is also reachable from $\{2\}$. From the initial state $n-1$, we reach the empty state under input $a_1$. Then $\cN_n''^D$ has $2^{n-3}+2$ reachable states. 

Since $\cN_n''^R$ is a subset of $\cD_n$ and it is deterministic, the DFA $\cN_n''^D$ is minimal. Therefore $\kappa(L_n''^R) = 2^{n-3}+2$, and it reaches the upper bound for reversal of both bifix- and factor-free regular languages with quotient complexity $n$~\cite{BJLS11}. \qed
\end{proof}

\section{Conclusions}\label{sec:cl}

Our results are summarized in Tables~\ref{tab:Summary1} and~\ref{tab:Summary2}. Each cell of Table~\ref{tab:Summary1} shows the syntactic complexity bounds of prefix- and suffix-free regular languages, in that order, with a particular alphabet size. Table~\ref{tab:Summary2} is structured similarly for bifix- and factor-free regular languages. The figures in bold type are tight bounds verified by {\it GAP}. To compute the bounds for suffix-, bifix-, and factor-free languages, we enumerated semigroups generated by elements of $\Bsf(n)$, $\Bbf(n)$, and $\Bff(n)$ that are contained in $\Bsf(n)$, $\Bbf(n)$, and $\Bff(n)$, respectively, and recorded the largest ones. By Propositions~\ref{prop:sf},~\ref{prop:bf},~\ref{prop:ff}, we obtained the desired bounds from the enumeration. The asterisk $\ast$ indicates that the bound is already tight for a smaller alphabet. In Table~\ref{tab:Summary1}, the last four rows include the tight upper bound $n^{n-2}$ for prefix-free languages, $\vsf(n)$, which is a tight upper bound for $2 \le n \le 5$ for suffix-free languages, conjectured upper bound $\wsf(n)$ for suffix-free languages, and a weaker upper bound $\bsf(n)$ for suffix-free languages. In Table~\ref{tab:Summary2}, the last four rows include $\vbf(n)$, which is a tight upper bound for bifix-free languages for $2 \le n \le 5$, conjectured upper bounds $\wbf(n)$ for bifix-free languages and $\wff(n)$ for factor-free languages, and weaker upper bounds $\bbf(n)$ for bifix-free languages and $\bff(n)$ for factor-free languages. 
\vspace{-.4cm}

\begin{table}[H]
\caption{Syntactic complexities of prefix- and suffix-free regular languages.}
\label{tab:Summary1}
\begin{center}
$
\begin{array}{|c||c|c|c|c|c|}    
\hline
\ \  \ \ &\ \ n=2 \ \ &\ \ n=3 \  \ & \  \  n=4 \ \ 
&  \ n=5 \ & \ n=6 \ \\
\hline \hline

|\Sig|=1
&	{\bf 1}	&	{\bf 2}	&	{\bf 3}	&	{\bf 4}	
&	{\bf 5}	\\
\hline

|\Sig|=2 
&	\ast	&	{\bf 3}/{\bf 3}		&	{\bf 11}/{\bf 11}	
&	{\bf 49}/{\bf 49}	&	? \\
\hline

|\Sig|=3
&	\ast	&	\ast				&	{\bf 14}/{\bf 13}
&	{\bf 95}/{\bf 61}	&	? \\
\hline

|\Sig|=4
&	\ast	&	\ast				&	\hspace{-.2cm}{\bf 16}/\ast 
&	\hspace{-.1cm}{\bf 110}/{\bf 67}	&	? \\
\hline

|\Sig|=5
&	\ast	&	\ast				&	\ast					
&	\hspace{-.1cm}{\bf 119}/{\bf 73}	&	? \\
\hline

|\Sig|=6
&	\ast	&	\ast				&	\ast					
&	\hspace{-.2cm}{\bf 125}/~\ast		&	~~?~/501 \\
\hline

|\Sig|=7
&	\ast	&	\ast				&	\ast					
&	\ast		&	\hspace{-.5cm}~{\bf 1296}/~?~ \\
\hline

|\Sig|=8
&	\ast	&	\ast				&	\ast					
&	\ast		&	~~\ast~/{\bf 629} \\
\hline

\cdots
&			&						&							
&						&    \\
\hline

~n^{n-2}~
&	\hspace{-.3cm}{\bf 1}			&	\hspace{-.3cm}{\bf 3}		&	\hspace{-.5cm}{\bf 16} 
&	\hspace{-.6cm}{\bf 125}			&	\hspace{-.9cm}{\bf 1296} \\
\hline

~\vsf(n)~
&\hspace{.3cm}{\bf 1}		&\hspace{.3cm}{\bf 3}	&\hspace{.5cm}{\bf 13}
&\hspace{.6cm}{\bf 73}		&\hspace{.8cm}501\\
\hline

~\wsf(n)~
&\hspace{.3cm}{\bf 1}	&\hspace{.3cm}{\bf 3}	&\hspace{.5cm}11
&\hspace{.6cm}67		&\hspace{.8cm}{\bf 629}\\
\hline

\bsf(n)
&\hspace{.3cm}{\bf 1}	&\hspace{.3cm}{\bf 3}	&\hspace{.5cm}15
&\hspace{.6cm}115		&\hspace{.8cm}1169\\
\hline
\end{array}
$
\end{center}
\label{table1}
\end{table}

\vspace{-1.4cm}

\begin{table}[H]
\caption{Syntactic complexities of bifix- and factor-free regular languages.}
\label{tab:Summary2}
\begin{center}
$
\begin{array}{|c||c|c|c|c|c|}    
\hline
\ \  \ \ &\ \ n = 2 \ \ & \ \ n=3 \  \ & \  \  n=4 \ \ 
&  \ n=5 \ & \ n=6 \ \\
\hline \hline

|\Sig|=1
&	{\bf 1}	&	{\bf 2}	&	{\bf 3}	
&	{\bf 4}	&	{\bf 5}	\\
\hline

|\Sig|=2 
&	\ast	&	\ast	&	{\bf 7}/{\bf 6} 
&	{\bf 20}/{\bf 12}	&	? \\
\hline

|\Sig|=3
&	\ast	&	\ast	&	\ast 
&	{\bf 31}/{\bf 16}	&	? \\
\hline

|\Sig|=4
&	\ast	&	\ast	&	\ast 
&	{\bf 32}/{\bf 19}	&	? \\
\hline

|\Sig|=5
&	\ast	&	\ast	&	\ast					
&	{\bf 33}/{\bf 20}	&	? \\
\hline

|\Sig|=6
&	\ast	&	\ast	&	\ast					
&	{\bf 34}/~?~		&	? \\
\hline

\cdots
&			&			&							
&			&				  \\
\hline

~\vbf(n)~
&\hspace{-.35cm}{\bf 1}	&\hspace{-.35cm}{\bf 2}	&\hspace{-.35cm}{\bf 7} 
&\hspace{-.6cm}{\bf 34}	&\hspace{-.6cm}209 \\
\hline

~\wbf(n)~
&\hspace{-.35cm}{\bf 1}	&\hspace{-.35cm}{\bf 2}	&\hspace{-.35cm}{\bf 7} 
&\hspace{-.6cm}33		&\hspace{-.6cm}{\bf 213} \\
\hline

~\wff(n)~
&\hspace{.35cm}{\bf 1}	&\hspace{.35cm}{\bf 2}	&\hspace{.35cm}{\bf 6}	
&\hspace{.5cm}{\bf 25}	&\hspace{.65cm}{\bf 150}  \\
\hline

\bbf(n)/\bff(n)
&{\bf 1}/{\bf 1}	&{\bf 2}/{\bf 2}	&{\bf 7}/{\bf 6}				
&41/31					&339/246	     \\
\hline

\end{array}
$
\end{center}
\label{table2}
\end{table}

\end{document}